\documentclass[11pt]{article}
\usepackage[T1]{fontenc}
\usepackage[utf8]{inputenc}
\usepackage[margin=1in]{geometry}
\usepackage{microtype}
\usepackage{amsmath,amssymb,amsthm}
\DeclareMathOperator*{\argmin}{arg\,min}
\usepackage{graphicx}
\usepackage{xifthen,enumitem,comment}
\usepackage[dvipsnames,svgnames,x11names]{xcolor}
\usepackage{etoolbox}
\usepackage{caption,subcaption}
\usepackage{booktabs,tikz}
\usepackage{array,delarray,framed}
\usepackage{tcolorbox}
\tcbset{paperbox/.style={frame empty}}
\usepackage[authoryear,round]{natbib}
\usepackage[hypertexnames=false]{hyperref}
\hypersetup{colorlinks=true,linkcolor=blue!50!black,citecolor=blue!50!black,urlcolor=blue!50!black,
  pdftitle={Pricing with an Informed Buyer},
  pdfauthor={Zhihao Gavin Tang, Yixin Tao, Shixin Wang}}
\usepackage{cleveref}
\crefname{claim}{claim}{claims}
\Crefname{claim}{Claim}{Claims}
\theoremstyle{plain}
\newtheorem{theorem}{Theorem}
\newtheorem{lemma}{Lemma}
\newtheorem{proposition}{Proposition}
\newtheorem{corollary}{Corollary}
\newtheorem{claim}{Claim}

\theoremstyle{definition}
\newtheorem{definition}{Definition}
\newtheorem{example}{Example}

\theoremstyle{remark}

\definecolor{myblue}{RGB}{0,0,200}
\ifdefined\pdfsuppresswarningpagegroup
\fi
\allowdisplaybreaks[1]

\setlist[enumerate,1]{label=\normalfont{(\roman*)},leftmargin=2em}
\makeatletter
\patchcmd{\env@cases}{1.2}{0.96}{}{}
\makeatother
\usetikzlibrary{matrix,decorations.pathreplacing}

\definecolor{longhorn}{rgb}{0.8, 0.33, 0.0}

\newcommand{\sF}[1][a]{\F^{-}_{\lambda,#1}}
\newcommand{\lF}[1][b]{\F^{+}_{\lambda,#1}}

\newcommand{\llF}{\F^{+}_{\lambda^\dagger, b^\dagger}}

\newcommand{\scF}{\cF^{-}}
\newcommand{\lcF}{\cF^{+}}

\newcommand{\barsF}[1][a]{\bar{\F}^{-}_{\lambda,#1}}
\newcommand{\barlF}[1][b]{\bar{\F}^{+}_{\lambda,#1}}
\newcommand{\barllF}{\bar{\F}^{+}_{\lambda^\dagger,b^\dagger}}

\newcommand{\scFone}{\cF^{-}|_{a=1}}
\newcommand{\lcFone}{\cF^{+}|_{b=1}}
\newcommand{\sFone}[1][a]{\F^{-}_{#1\lambda, 1}}
\newcommand{\lFone}[1][b]{\F^{+}_{#1\lambda, 1}}

\newcommand{\lb}{\mathsf{LB}(\sF,\Gamma)}
\newcommand{\ub}{\mathsf{UB}(\lF,\Gamma)}

\newcommand{\distle}{\preceq_{1}}
\newcommand{\distge}{\succeq_{1}}

\renewcommand{\Pr}{\mathbb{P}}
\renewcommand{\P}{\mathbb{P}}

\providecommand{\E}{\mathbb{E}}
\providecommand{\R}{\mathbb{R}}
\providecommand{\F}{\mathbb{F}}
\def\G{\mathbb{G}}

\providecommand{\cE}{\mathcal{E}}
\providecommand{\cF}{\mathcal{F}}

\newcommand{\of}[1]{\left(#1\right)}
\newcommand{\off}[1]{\left[#1\right]}
\newcommand{\offf}[1]{\left\{#1\right\}}

\newcommand{\barf}[1][\F]{\overline{#1}}

\newcommand{\beqn}{\begin{eqnarray*}}
\newcommand{\eeqn}{\end{eqnarray*}}
\newcommand{\beq}{\begin{eqnarray}}
\newcommand{\eeq}{\end{eqnarray}}

\newcommand{\ba}{\begin{aligned}}
\newcommand{\ea}{\end{aligned}}

\providecommand{\Rev}{\mathrm{Rev}}

\providecommand{\OPT}{\mathrm{OPT}}
\newcommand{\opt}{\mathrm{OPT}}

\newcommand{\hide}[1]{}

\newcommand{\dd}[1]{\;\textrm{d}#1}

\newcommand{\Ex}[2]{\underset{{#1}}{\E}\off{{#2}}}
\renewcommand{\Pr}[2][]{\underset{{#1}}{\P}\off{#2}}

\title{Pricing with an Informed Buyer}
\author{Zhihao Gavin Tang\thanks{Key Laboratory of Interdisciplinary Research of Computation and Economics, Shanghai University of Finance and Economics. Email: \texttt{tang.zhihao@mail.shufe.edu.cn}.}
\and Yixin Tao\thanks{Key Laboratory of Interdisciplinary Research of Computation and Economics, Shanghai University of Finance and Economics. Email: \texttt{taoyixin@mail.shufe.edu.cn}.}
\and Shixin Wang\thanks{H. Milton Stewart School of Industrial and Systems Engineering, Georgia Institute of Technology. Email: \texttt{shixin.wang@isye.gatech.edu}.}}
\date{}

\begin{document}
\maketitle
\begin{abstract}

We study prior-independent pricing for a single seller and a single buyer when the seller has access to only one sample from the buyer's value distribution, whereas the buyer knows the distribution. Existing robust-pricing approaches either rely on distributional statistics, which may require many samples to estimate, or use realized samples to determine a price or mechanism before the buyer decides whether to purchase. We bridge these regimes by exploiting the buyer's informational advantage: policies that ordinarily require the seller to know a statistic such as the mean, an $L^\eta$-norm, or a superquantile can be implemented with a single private sample.

We introduce sample-contingent mechanisms: the seller commits ex ante to a transfer rule, and the buyer participates and reports before observing the sample. Subject to standard regularity conditions, every concave pricing policy is implementable in this way. To evaluate performance, we reduce the worst-case analysis of any monotone pricing policy over $\alpha$-regular distributions to two parametric boundary families. This reduction allows us to characterize the optimal monotone sample-contingent mechanism; for monotone hazard rate (MHR) distributions, it attains an approximation ratio between $0.791$ and $0.796$. We also establish impossibility results for concave pricing policies. Finally, the same framework applies to statistic-based robust pricing, unifying the statistic-based and sample-based approaches.

\end{abstract}

\newpage
\section{Introduction}

In many thin markets, the seller is both data-poor and less experienced than the buyer. A university lab licensing a patent may know little about its commercial value, whereas a pharmaceutical or technology company has experience with comparable licenses and downstream demand. Similarly, an heir selling an artwork may know less than a professional dealer who observes comparable sales and resale outcomes. A job candidate negotiating with a private technology firm may know only a broad compensation range, whereas the firm repeatedly negotiates with similar candidates. These examples raise a basic question: how should an informationally disadvantaged seller price an item for a professional buyer who understands the market better?

We study this question in the canonical single-item, single-buyer pricing problem. Following the prior-independent pricing literature, we model the seller's limited market knowledge parsimoniously: the seller has access to only one sample from the unknown value distribution. In the examples above, this sample can be interpreted as a comparable licensing outcome, an auction outcome for a similar artwork, an independent appraisal, a compensation benchmark, a competing offer, or another credible market signal.

The standard sample-based approach treats the sample solely as the seller’s source of information. The seller observes one or more samples from the unknown distribution and then uses them to set a price. This approach is natural and has generated a rich literature on prior-independent pricing and mechanism design. However, it misses an important feature of the motivating examples. In these markets, the seller’s problem is not merely that she has few samples. Rather, the decision-relevant market knowledge is held, at least in part, by the buyer. A dealer, licensee, employer, or intermediary may have a much better understanding of the distribution from which the relevant market signal is drawn. Much of the standard sample-based pricing literature uses realized samples to determine the price or mechanism before the buyer's purchase decision, without leveraging the buyer's prior knowledge about the distribution of that sample.

Our idea is to embed the sample directly into the pricing mechanism. The seller commits ex ante to a payment rule that depends on the sample and, possibly, on a report from the buyer. The buyer then decides whether to participate before observing the sample. This information structure can arise in two ways: the seller may already have drawn the sample and concealed it from the buyer, or the sample may be a delayed signal, realized publicly only after the buyer participates and reports. We refer to either case as a \emph{private sample}. Although the buyer does not know its realization when deciding whether to participate, she knows the distribution from which it is drawn. Her decision therefore depends on the expected payment under the committed rule. By designing that rule appropriately, the seller can make the buyer's decision depend on a desired statistic of the value distribution without knowing the statistic himself.

Thus, the mechanism uses the buyer’s market knowledge without first asking the buyer to report a distribution and then solving a separate pricing problem. Information elicitation and pricing optimization occur in one step. The sample is not used merely for ex post estimation; it becomes part of the economic mechanism itself. In this sense, our work develops an ``elicit-and-optimize'' approach to pricing with limited seller information, in which the informed buyer is both the source of information and the strategic participant endogenous to the downstream price optimization.

We call such mechanisms \emph{sample-contingent mechanisms}. In a sample-contingent mechanism, the seller commits before the buyer’s participation decision to a transfer rule\ that maps the sample, and possibly a buyer report, into a payment. With an appropriately designed rule, the seller can implement a broad class of statistic-based pricing policies using only one sample. For instance, although estimating the mean of a distribution to high accuracy generally requires many samples, a seller with only a single sample can implement the same pricing rule as if the mean were known, provided the buyer has the relevant distributional knowledge. We illustrate the intuition in the following example.

\begin{tcolorbox}[paperbox,
    top=2pt,
    bottom=2pt,
     before upper={
     \setlength{\baselineskip}{16pt}
        \setlength{\parskip}{6pt} 
    }]
\paragraph{Mean pricing.} 
The seller commits to a transfer rule
\[
h(s) = s~,
\]
where $s$ is a private sample. If the buyer accepts the offer, she pays $h(s)=s$. Because she does not observe $s$ when deciding whether to participate, she compares her value $v$ with the expected price
\[
p_h(\F) = \Ex{s\sim\F}{h(s)} = \Ex{s\sim\F}{s} = \|\F\|_1~.\footnote{We use $\|\F\|_\eta$ to denote $\E_{s \sim \F}[s^{\eta}]^{1/\eta}$,}
\]
where $\F$ is the buyer's value distribution. Hence, the buyer accepts the offer if and only if her valuation $v\ge \|\F\|_1$, and pays $\|\F\|_1$ in expectation; otherwise she declines to participate. Therefore, this mechanism is effectively equivalent to posting a price equal to the mean of the distribution.
\end{tcolorbox}

In this example, the buyer’s knowledge is only used implicitly. If the buyer does not know the mean of the distribution, then no natural solution concept dictates how she would behave under our sample-contingent mechanism. On the other hand, she does not need full knowledge of the entire distribution---knowing only the mean is sufficient for making a rational participation decision. The seller may also adopt a pricing policy that posts a discounted multiple of the mean, which is an effective pricing strategy in the robust pricing regime when the mean is known.

Moreover, the seller can explicitly involve the buyer in the pricing process, effectively letting her ``price for herself''. We illustrate this idea with two examples: $L^2$-norm pricing and superquantile pricing.

\begin{tcolorbox}[paperbox,
    top=2pt,
    bottom=2pt,
     before upper=
     {\setlength{\baselineskip}{16pt}
        \setlength{\parskip}{6pt}
    }]
\paragraph{$L^2$-norm pricing.} 
The seller commits to a transfer rule indexed by $x$
\[
h(s,x) = 0.5  \left( x + \frac{s^2}{x} \right)~,
\]
where $s$ is a private sample and $x$ is the report solicited from the buyer.
A buyer with any private valuation $v$ minimizes her expected payment $\E[h(s,x)]$ by reporting $x = \|\F\|_2 = \sqrt{\E[s^2]}$. Under this report, the expected payment becomes 
\[
p_h(\F) = \Ex{s\sim\F}{h(s, \|\F\|_2)} = \sqrt{\E[s^2]} = \|\F\|_2~.
\]
Hence, the buyer accepts the offer if and only if her valuation $v \ge \|\F\|_2$, and pays $\|\F\|_2$ in expectation; otherwise she declines to participate.
Therefore, this mechanism is effectively equivalent to posting a price equal to the $L^2$-norm of the distribution.
\end{tcolorbox}
We remark that this sample-contingent mechanism\ only requires the buyer to know the $L^2$-norm of the distribution, which allows her to determine both what to report and whether to participate. The same construction extends naturally to general $L^\eta$-norms for any $\eta \ge 1$. 
The following example illustrates the use of the superquantile in a pricing policy. For \(q\in(0,1]\), define Value at Risk directly by
\[
\mathrm{VaR}_{q}(\F):=\sup\{z\ge0:\barf(z)\ge q\},
\]
and define the superquantile by
\[
\mathrm{CVaR}_{q}(\F)
:=\frac{1}{q}\int_0^q \mathrm{VaR}_{u}(\F)\,\dd u
=\inf_{x\ge0}\left\{x+\frac{1}{q}\Ex{s\sim\F}{(s-x)^+}\right\}.
\]
Here and throughout, \(\barf(z):=\Pr[v\sim\F]{v\ge z}\) denotes the weak survival function.\footnote{In the finance and optimization literature, the superquantile is commonly referred to as the Conditional Value at Risk (CVaR), and the quantile is commonly referred to as the Value at Risk (VaR). We adopt this notation in the present paper. Unlike the risk-measure interpretation, we view CVaR here solely as a distributional statistic.} 

\begin{tcolorbox}[paperbox,
    top=2pt,
    bottom=2pt,
     before upper={\setlength{\baselineskip}{16pt}
        \setlength{\parskip}{6pt}
    }]
\paragraph{Superquantile pricing.}  The seller announces publicly a $q \in (0,1]$, and commits to a transfer rule
\[
h(s,x) = x +  \frac{1}{q} \cdot (s- x)^+~,\footnote{We use $x^+$ to denote the function $\max(x,0)$.}
\]
where $s$ is a private sample and $x$ is the report solicited from the buyer.
A buyer with any private valuation $v$ minimizes her expected payment $\E[h(s,x)]$ by reporting any \(x\) satisfying \(\Pr[s\sim\F]{s>x}\le q\le\Pr[s\sim\F]{s\ge x}\). Every such report is a minimizing \(q\)-quantile, and the minimum expected payment is
\[
p_h(\F)=\inf_{x\ge0}\Ex{s\sim\F}{h(s,x)}=\mathrm{CVaR}_{q}(\F)~.
\]
Hence, the buyer accepts the offer if and only if her valuation $v \ge \mathrm{CVaR}_{q}(\F)$, and pays $\mathrm{CVaR}_{q}(\F)$ in expectation; otherwise she declines to participate.
Therefore, this mechanism is effectively equivalent to posting a price equal to the $\mathrm{CVaR}_{q}$ of the distribution.
\end{tcolorbox}
In this example, our mechanism requires the buyer to know the \(q\)-quantile\ in order to determine what to report, and $\mathrm{CVaR}_{q}$ to decide whether to participate.

As we have seen, the seller can implement a rich family of pricing policies $p:\Delta(\R_+)\to\R$ by employing a carefully designed transfer rule \(h\)\ together with a single sample. 
When the buyer has full knowledge of the distribution, the transfer rule \(h\)\ functions as a \emph{proper scoring rule}: it maps the sample $s$ and the buyer's reported distribution $\F$ to a payment. 
In our context, a transfer rule\ is \emph{proper} if reporting the distribution truthfully minimizes the expected payment.
By the Savage representation theorem for proper scoring rules, a concave pricing policy \(p\) satisfying the regularity conditions in \Cref{thm:scoring} can be implemented through such a sample-contingent mechanism.

We adopt the prior-independent mechanism design framework to evaluate the performance guarantees of different pricing policies and to identify the optimal sample-contingent mechanism. Specifically, we assume the seller has access to a single sample and knows only that the buyer’s value distribution lies in a nonparametric ambiguity set of ``$\alpha$-regular'' distributions, an assumption widely used in the literature. The case $\alpha=1$ corresponds to the class of monotone hazard rate (MHR) distributions, which includes a broad class of distributions such as uniform, truncated normal, exponential, logistic, and extreme value distributions. The seller aims to find the optimal pricing policy hedging against the most adversarial distribution that nature may choose from ``$\alpha$-regular'' distributions. We adopt the maximin ratio objective, in which the performance of a pricing policy is measured relative to the hindsight-optimal expected revenue that a fully informed clairvoyant seller could obtain. 
A central challenge is that evaluating the worst-case approximation ratio of any pricing policy requires solving an infinite-dimensional optimization problem over a generally nonconvex set of distributions. 

We overcome this difficulty by restricting attention to pricing policies that are monotone with respect to first-order stochastic dominance---a natural property satisfied by all previously considered examples. Our main contribution is a tractable reduction of nature's optimization problem. For any deterministic monotone pricing policy, we show that the worst-case infimum is unchanged when the ambiguity set is restricted to a simple two-parameter boundary family. For policies based on a specific statistic---such as the mean, the $L^2$-norm, or the superquantile---the infimum can be evaluated over a one-parameter family. %

The reduction distinguishes overpricing from underpricing relative to the monopoly price. If the seller overprices, a stochastically larger valuation distribution yields a weakly lower performance ratio; if the seller underprices, a stochastically smaller distribution does so. This observation reduces the infinite-dimensional ambiguity set to a tractable parametric class.

Beyond quantifying the approximation ratios of specific pricing policies, our reduction enables us to characterize the optimal monotone sample-contingent mechanism---or equivalently, the optimal monotone concave pricing policy---for every $\alpha \in [0,1]$. With computer assistance, we compute that the optimal monotone sample-contingent mechanism\ attains an approximation ratio of $\sim 0.79$ for MHR distributions.

\Cref{tab:ratio-policy} summarizes the optimal discounts and performance ratios under MHR distributions for several representative pricing rules. Notably, simple rules such as mean pricing and superquantile pricing deliver performance guarantees close to that of the optimal monotone sample-contingent mechanism. %
\begin{table}[htbp]
\centering
\caption{Optimal Approximation Ratios with Different Pricing Policies over MHR Distributions}
\label{tab:ratio-policy}
\begin{tabular}{lcccc}
\toprule
Pricing Policy $p$
& Transfer Rule \(h\)
& Parameter
& Discount $\omega$ 
& Ratio \\
\midrule
Mean 
& $\omega \cdot s$
& - 
& 0.823
& 0.771 \\
\midrule
$L^\eta$-norm 
& $\omega \cdot \frac{1}{\eta} \left( (\eta-1)\|\F\|_{\eta} + \frac{s^\eta}{ \|\F\|_{\eta}^{\eta-1}} \right) $
& $\eta = 1.37$
& 0.805
& 0.774 \\
\midrule
Superquantile
& $\omega \cdot \left( \mathrm{VaR}_{q}(\F) + \frac{1}{q} (s-\mathrm{VaR}_{q}(\F))^+ \right)$
& $q=0.92$
& 0.787
& 0.787 \\
\midrule
Optimal Monotone & Theorem~\ref{thm:optimal}  & - & - & $\sim0.791$ \\
\bottomrule
\end{tabular}
\end{table}  

Our reduction framework not only enables the development of positive results (lower bounds) for a broad class of pricing rules, but also provides a foundation for establishing negative results (upper bounds). We show that no concave pricing rule can achieve an approximation ratio above 0.801 under MHR distributions, indicating that restricting attention to monotone pricing rules leads to only a negligible performance loss. 

Finally, beyond sample-based settings, our framework provides a unified analysis of robust pricing problems in which the seller observes any monotone and homogeneous\ statistic $\Psi(\F)$, such as mean, quantile, or lower-tail expectation, even when the corresponding functional is not concave. In particular, the optimal mechanisms and approximation ratios presented in \Cref{tab:ratio-policy} apply directly to robust pricing problems subject to mean, $L^\eta$-norm, or superquantile information, respectively. The equivalence between statistic-based and sample-based implementations is discussed in detail in \Cref{sec:application}, which establishes that our approach provides a unified framework that bridges the statistic-based and sample-based robust pricing regimes.

\paragraph{Roadmap.} In the next subsection, we discuss related work. Section~\ref{sec:model} develops sample-contingent mechanisms, their induced interim prices, and their relation to other one-sample paradigms. Section~\ref{sec:monotone} formulates robust performance, introduces the \(\alpha\)-regular ambiguity set and its boundary distributions, and proves our worst-case reduction for monotone pricing policies.\ We characterize the optimal monotone sample-contingent mechanism\ in Section~\ref{sec:optimal}. Section~\ref{sec:negative} establishes our negative results, and Section~\ref{sec:application} explores the connection between sample-contingent mechanisms\ and robust pricing with statistic information. 

\subsection{Related Work}
Prior-independent pricing and mechanism design have proceeded along two distinct paths, distinguished by the type of information available to the seller. 

\paragraph{Sample-based pricing.} The first line of research, which we call the sample-based model, designs pricing and auction mechanisms using samples from the unknown value distribution. %
In the one-sample setting, \cite{SIAM/HuangMR18} show the optimal deterministic approximation ratio under regular distribution is $0.5$, and they further establish a 0.589 lower bound and a 0.68 upper bound for deterministic mechanisms under MHR distributions. \cite{or/AllouahBB22} improve the MHR lower bound to 0.644 and provide an upper bound of 0.648 for randomized mechanisms.
\cite{EC/FuILS15} are the first to demonstrate the benefit of randomized mechanisms. Later, \cite{or/AllouahBB22} improve the lower bound of randomized mechanisms to $0.502$ under regular distributions.
With two samples, \cite{EC/BabaioffGMM18,EC/DaskalakisZ20} derive lower bounds for deterministic mechanisms under regular distributions. \cite{or/AllouahBB22} provides a general approach to develop lower bounds for a general number of samples. 
A parallel stream of literature investigates the sample complexity to achieve given approximation ratios \citep{cole2014sample,gonczarowski2017efficient,SIAM/HuangMR18, guo2019settling,hu2021targeting}.
Existing sample-based pricing work uses realized samples to determine the price or mechanism before the buyer's purchase decision~\citep{EC/DhangwatnotaiRY10,SIAM/HuangMR18,or/AllouahBB22}. The final subsection of \Cref{sec:model} explains why the public-sample benchmark is not directly comparable to our setting: we require the buyer to participate and report before observing the sample.

The closest work to ours is \cite{stoc/FengHL21}, which shares the key timing that the buyer knows \(\F\) but not the sample realization; its sample-bid mechanism lets the buyer's bid affect both allocation and payment. For MHR distributions, their equilibrium outcome is equivalent to discounted-mean pricing \(h(s)=\beta s\). Outside MHR distributions, interior bids may be optimal, so neither mechanism class should be read as containing the other; the final subsection of \Cref{sec:model} gives the precise comparison.

\paragraph{Statistic-based pricing.} The second line of research, which we term the ``statistic-based" model, considers settings where the seller possesses partial information about the value distribution---such as moments, quantiles, or other aggregate statistics. The seller’s objective is to choose a price or mechanism that maximizes performance against the most pessimistic distribution consistent with this information. Robust mechanisms are studied under the maximin revenue objective for ambiguity sets based on moments information \citep{pinar2017robust,carrasco2018optimal,chen2022distribution,chen2024screening}, mean preserving contraction \citep{du2018robust,chen2024screening} and Wasserstein metric \citep{li2019revenue,chen2024screening}. \cite{bergemann2008pricing,caldentey2017intertemporal,koccyiugit2022robust} study robust mechanisms with support information under the minimax regret objective. 
 For the maximin ratio objective, robust mechanisms are derived for support \citep{eren2010monopoly}, moments \citep{giannakopoulos2023robust, wang2024minimax}, and quantile \citep{ms/AllouahBB23, bahamou2024fast,wang2025power, ec/NabyFH25} ambiguity sets. 

Among these works, \cite{ms/AllouahBB23} is closest to our setting because it studies quantile-based ambiguity sets for \(\alpha\)-regular distributions. A quantile observation fixes a local point on the demand curve, whereas a general statistic such as a mean, moment, or superquantile imposes a global distributional constraint; consequently, their single-point extremal construction does not directly extend to our setting. \Cref{sec:application} gives the detailed comparison and explains our route through the statistic-based equivalence.

\paragraph{Scoring rules.} The study of scoring rules dates back to \cite{brier1950verification}, with \cite{savage1971elicitation} providing a comprehensive foundation for evaluating probabilistic forecasts. We refer the reader to \cite{sigecom/LambertPS08} and \cite{sigecom/FrongilloW24} for broader surveys.
A large literature examines the elicitability of statistical properties~\citep{osband1985providing,gneiting2007strictly,ec/LambertPS08}. In this framework, an agent reports a property (e.g., the mean or median), and a scoring rule is proper if truthful reporting maximizes expected score. A key result is that CVaR is not elicitable by itself, whereas the pair (VaR, CVaR) is jointly elicitable~\citep{fissler2016higher}.

Our work, however, takes a different perspective: rather than designing them to elicit a statistical property, we design them so that their Bayes risk---the optimized expected transfer---implements the desired distribution-dependent price. This distinction is important for interpreting our examples. For instance, in our $L^\eta$-norm example, the scoring rule both elicits the $L^\eta$-norm and has a Bayes risk equal to it---an incidental coincidence. In our CVaR-pricing example, the scoring rule instead elicits VaR while its Bayes risk equals CVaR, which is the designed objective. Thus, our use of scoring rules concerns shaping the Bayes risk, rather than determining whether a property (such as CVaR) is directly elicitable. Moreover, the optimal transfer rule\ in Section~\ref{sec:optimal} requires the buyer to report the full distribution, placing our setting conceptually outside the standard property-elicitation framework.

\cite{caillaud2005implementation} study an auction setting in which the seller does not know the buyers’ value distributions. \cite{itcs/ACM12} investigate a related setting in which “buyers know each other better than the seller knows them.” Both papers share our premise that buyers may possess better market information than the seller. In particular, \cite{itcs/ACM12} incorporate scoring rules into multi-agent mechanism design. A key difference, however, is that in their framework the elicited information affects only the seller’s strategy toward the other buyers, thereby decoupling information elicitation from the mechanism design problem. In contrast, in our setting the pricing mechanism itself is a scoring rule, which we optimize directly.

\section{\texorpdfstring{Sample-Contingent Mechanisms}{Sample-Contingent Mechanisms}}
\label{sec:model}

Our main conceptual contribution is the sample-contingent mechanism itself. It gives the seller a systematic way to use the buyer's knowledge of the value distribution while requiring only one private sample from that distribution. Rather than estimating \(\F\), the seller commits ex ante to a transfer \(h(s,\F')\) that depends on the sample $s$ and the buyer's report $\F'$. Because participation and reporting occur before the sample is observed, the buyer evaluates this random transfer through its expectation under her known distribution; a proper transfer rule can therefore elicit that knowledge and implement a distribution-dependent interim price. This construction is independent of the robust optimization developed later. Robust optimization selects and evaluates induced pricing policies over an ambiguity set, whereas the sample-contingent architecture is the mechanism-design device that makes a rich class of such policies implementable by an uninformed seller.

\subsection{\texorpdfstring{Transfer Rules and Induced Interim Prices}{Transfer Rules and Induced Interim Prices}}
We study a robust monopoly-pricing problem in which a seller offers one product to one buyer. The buyer has a private value \(v\in\R_+\), drawn from an unknown distribution \(\F\), and an independent private sample \(s\sim\F\) is available for contingent settlement. The buyer knows her realized value \(v\) and the distribution \(\F\), whereas the seller knows neither. The seller knows only that \(\F\) belongs to a distribution class \(\cF\).

The timing is as follows:
\begin{itemize}
\item The seller publicly commits to a transfer rule $h$ before the buyer's participation and report decisions.
\item Nature selects \(\F\in\cF\) and draws \(v\sim\F\) and an independent private sample \(s\sim\F\).
\item The buyer privately observes \((v,\F)\).
\item The buyer decides whether to participate and, if she participates, reports a distribution \(\F'\in\Delta(\R_+)\).
\item The private sample is used to settle the transfer. If the buyer participated, she receives the item and pays \(h(s,\F')\); otherwise, the item is not allocated and no transfer occurs.
\end{itemize}

A \emph{sample-contingent mechanism} is specified by a transfer rule
\[
    h:\R_+\times\Delta(\R_+)\to\R,
\]
where $h(s,\F')$ is the transfer from the buyer to the seller when the realized sample is $s$ and the buyer reports the distribution $\F'$. 
We allow the realized transfer to be signed at this general level. The mechanisms will ultimately be evaluated through the interim prices they induce.

A risk-neutral buyer who believes that the true distribution is $\F$ chooses a report that minimizes her expected payment.
Accordingly, the \emph{induced interim price} associated with $h$ is
\[
p_h(\F) := \inf_{\F'} \Ex{s \sim\F}{h(s,\F')}.
\]
Although \(p_h(\F)\) is indexed by the true distribution \(\F\), the seller neither observes \(\F\) nor computes this quantity. The seller commits only to the transfer rule \(h\); \(p_h(\F)\) is the buyer's equilibrium expected payment before the sample is observed.
We assume throughout that this infimum is attained.

Recall that \(\barf(x)=\Pr[v\sim\F]{v\ge x}\) denotes the weak survival function of \(\F\).
Because participation results in certain allocation of the product, a buyer with type $(v,\F)$ purchases if and only if 
$v\ge p_h(\F)$,
where we adopt seller-favorable tie-breaking. The corresponding selling probability is
\[
\barf\bigl(p_h(\F)\bigr)
=
\Pr[v \sim \F]{v \ge p_h(\F)}~.
\]
Thus, from the buyer's interim perspective, the sample-contingent mechanism is characterized by the distribution-dependent price $p_h(\F)$.
For a finite scalar posted price \(z\), define its revenue by
\[
\Rev(z,\F):=z\,\barf(z).
\]
For a generic distribution-dependent pricing policy \(p\), and for the interim price induced by \(h\), we use the overloaded notation
\[
\Rev(p,\F):=\Rev\bigl(p(\F),\F\bigr),
\qquad
\Rev(p_h,\F):=\Rev\bigl(p_h(\F),\F\bigr).
\]

\subsection{\texorpdfstring{Proper Transfer Rules and Concave Interim Prices}{Proper Transfer Rules and Concave Interim Prices}}

An important consequence is that every sample-contingent mechanism induces a concave interim price. 
\begin{definition}[Concave Pricing Policies]
A pricing policy $p:\Delta(\R_+) \to \R$ is concave if for all $\F_1,\F_2 \in \Delta(\R_+)$ and all $\lambda \in [0,1]$,
\[
p\!\left(\lambda \F_1 + (1-\lambda)\F_2\right) \ge \lambda\, p(\F_1) + (1-\lambda)\, p(\F_2),
\]
where $\lambda \F_1 + (1-\lambda)\F_2$ denotes the mixture distribution.
\end{definition}

\begin{lemma}
    For any transfer rule $h:\R_+ \times \Delta(\R_+) \to \R$, the induced interim price $p_h$ is concave.
\end{lemma}
\begin{proof}
For any $\F_1,\F_2\in\Delta(\R_+)$ and $\lambda\in[0,1]$,
\begin{align*}
p_h\bigl(\lambda\F_1+(1-\lambda)\F_2\bigr)
&=
\inf_{\F'}
\left\{
\lambda\Ex{s\sim\F_1}{h(s,\F')}
+(1-\lambda)\Ex{s\sim\F_2}{h(s,\F')}
\right\} \\
&\ge
\lambda\inf_{\F'}\Ex{s\sim\F_1}{h(s,\F')}
+(1-\lambda)\inf_{\F'}\Ex{s\sim\F_2}{h(s,\F')} \\
&=
\lambda p_h(\F_1)+(1-\lambda)p_h(\F_2).
\end{align*}
For each fixed report $\F'$, the expected transfer is affine in the true distribution $\F$. The induced interim price is therefore the lower envelope of a family of affine functionals and is consequently concave.

\end{proof}

\begin{definition}[Proper Transfer Rules]
A transfer rule $h: \R_{+} \times \Delta(\R_{+}) \to \R$ is proper if:
\[
\Ex{s \sim \F}{h(s,\F)} \le \Ex{s \sim \F}{h(s,\F')}, \quad \forall \F,\F'~.
\]
\end{definition}
Under a proper transfer rule, the buyer minimizes her expected payment by truthfully reporting the underlying distribution. The term \emph{proper} is standard in the information-elicitation literature, where it refers to rules that incentivize an agent to report her beliefs truthfully.

We have shown that every sample-contingent mechanism induces a concave interim price. The classical representation of proper scoring rules due to \cite{savage1971elicitation} and \cite{gneiting2007strictly} establishes the reverse direction: subject to the regularity conditions stated below, every concave pricing policy can be implemented by a proper sample-contingent transfer rule.

\begin{theorem}[\cite{savage1971elicitation,gneiting2007strictly}]
\label{thm:scoring}
Subject to the integrable-supergradient condition stated below, a pricing policy $p:\Delta(\R_{+}) \to \R$ is implementable by a proper sample-contingent transfer rule if and only if $p$ is concave. 
In particular, suppose that $p$ admits an integrable supergradient $g_{\F}$ at every $\F$, so that
\[
p(\F') \le p(\F) + \left\langle g_{\F},\F'-\F\right\rangle
\qquad
\text{for all }\F,\F'\in\Delta(\R_+).
\]
Then the transfer rule 
\[
h_p(s,\F) := p(\F) + \left\langle g_{\F}, \, \delta_s - \F \right\rangle,
\]
where $\delta_s$ denotes the Dirac measure at $s$, is proper and satisfies
\[
p(\F) = \Ex{s\sim \F}{h_p(s,\F)} = \inf_{\F'} \Ex{s \sim \F}{h_p(s,\F')}~.
\]
\end{theorem}

Thus, subject to the theorem's stated regularity assumptions, concave distribution-dependent interim prices are exactly the reduced forms implementable by proper sample-contingent transfer rules; for a risk-neutral buyer, mechanism design can therefore be studied directly over these interim prices.

\begin{example}
We revisit the three pricing policies from the introduction and give proper transfer rules that implement them.
\begin{enumerate}
    \item \emph{Mean pricing.}
    The mean \(p(\F)=\|\F\|_1=\E[s]\) is affine, and hence concave, in \(\F\).
    It is implemented by the proper transfer rule \(h(s,\F)=s\).
    Since the transfer does not depend on the report, every report yields the same expected payment, and truthful reporting is weakly optimal.

    \item \emph{\(L^\eta\)-norm pricing.}
    For \(\eta>1\), the functional \(p(\F)=\|\F\|_\eta=(\E[s^\eta])^{1/\eta}\) is concave in \(\F\).
    It is implemented by the proper transfer rule
    \(h(s,\F)=\frac{1}{\eta}\bigl((\eta-1)\|\F\|_\eta+s^\eta/\|\F\|_\eta^{\eta-1}\bigr)\).
    The case \(\eta=1\) reduces to mean pricing.

    \item \emph{Superquantile pricing.}
    For \(q\in(0,1]\), the superquantile is concave in \(\F\), since
    \(\mathrm{CVaR}_q(\F)=\inf_{x\ge0}\{x+\frac{1}{q}\E[(s-x)^+]\}\)
    is the infimum of affine functionals of \(\F\).
    It is implemented by the proper transfer rule
    \(h(s,\F)=\mathrm{VaR}_q(\F)+\frac{1}{q}(s-\mathrm{VaR}_q(\F))^+\).
    The case \(q=1\) also reduces to mean pricing.
\end{enumerate}
\end{example}
Although Theorem~\ref{thm:scoring} is formulated in terms of a full distribution report, these examples show that such a rich report is not always necessary. 
For mean pricing, no report is needed; for \(L^\eta\)-norm pricing, the transfer depends on the reported distribution only through its \(L^\eta\)-norm; and for superquantile pricing, only a minimizing \(q\)-quantile is needed. 
Thus, as illustrated in the introduction, the distribution report can be reduced to a scalar report \(x\) without changing the induced interim price. 
This simplification is specific to these policies; a general concave pricing policy may still require the buyer to report the full distribution.

\subsection{\texorpdfstring{Relation to Other One-Sample Paradigms}{Relation to Other One-Sample Paradigms}}

In the standard public-sample model, the buyer observes the realized sample before deciding whether to participate~\citep{EC/DhangwatnotaiRY10,SIAM/HuangMR18,or/AllouahBB22}. A mechanism draws a possibly randomized scalar, multiplies it by the sample, and posts the result as the price. This benchmark is not directly comparable to ours because we require the buyer to participate and report before observing the sample. Under our timing, the buyer evaluates the transfer rule in expectation and faces the induced interim price \(p_h(\F)=\inf_{\F'}\Ex{s\sim\F}{h(s,\F')}\).

\cite{stoc/FengHL21} study the same key informational asymmetry---the buyer knows \(\F\) but not the sample realization---while allowing her bid to affect both allocation and payment. Their sample-bid mechanism asks the buyer to bid \(b\ge0\) without observing \(s\), allocates the item if and only if \(b\ge s\), and charges \(\beta\min\{b,s\}\). For MHR distributions, they show that the buyer's optimal bid is
\[
b(v,\F)
=
\begin{cases}
0, & v<\beta\Ex{s\sim\F}{s},\\
+\infty, & v\ge\beta\Ex{s\sim\F}{s}.
\end{cases}
\]
Low types' zero bid corresponds to declining participation, whereas high types receive the item surely and pay \(\beta s\). Consequently, under MHR distributions the equilibrium outcome of sample-bid is equivalent to our discounted-mean sample-contingent mechanism \(h(s)=\beta s\), whose induced interim price is \(p_h(\F)=\beta\Ex{s\sim\F}{s}\). Outside MHR distributions, interior bids may be optimal, so sample-bid generally does not reduce to our sure-allocation sample-contingent mechanism class.

\section{{Robust Performance and Worst-Case Distributions}}
\label{sec:monotone}

\subsection{Robust Objective and the \texorpdfstring{$\alpha$}{alpha}-Regular Ambiguity Set}
\label{sec:prelim}

For a valuation distribution $\F$, let
\[
\opt(\F):=\sup_{z\in\R_+}\Rev(z,\F)
\]
denote the optimal posted-price revenue. Throughout, we restrict attention to ambiguity sets $\cF$ such that
$0<\opt(\F)<+\infty$ for every $\F\in\cF$. The worst-case approximation ratio of a pricing policy $p$ over $\cF$ is
\[
\Gamma(p,\cF):=\inf_{\F\in\cF}\frac{\Rev(p,\F)}{\opt(\F)}~.
\]
The robust-pricing objective is to maximize $\Gamma(p,\cF)$ over the specified class of pricing policies.

To accommodate the possible nonattainment of $\opt(\F)$ at any finite price, let
$\overline{\R}_+:=\R_+\cup\{+\infty\}=[0,+\infty]$. For every $\alpha$-regular distribution considered below, regularity implies that $z\mapsto\Rev(z,\F)$ is eventually monotone on $\R_+$. Since
$0\le \Rev(z,\F)\le\opt(\F)<+\infty$, the limit exists, and we extend the revenue function to $\overline{\R}_+$ by
$\Rev(+\infty,\F):=\lim_{z\to+\infty}\Rev(z,\F).$
With this convention,
\[
\opt(\F)=\max_{z\in\overline{\R}_+}\Rev(z,\F),
\qquad
p^*(\F):=
\min\bigl\{z\in\overline{\R}_+:
\Rev(z,\F)=\opt(\F)\bigr\}.
\]
Thus, $p^*(\F)$ is the leftmost revenue-maximizing price in the extended price space. If no finite price attains $\opt(\F)$, then $p^*(\F)=+\infty$ and $\Rev(+\infty,\F)=\opt(\F)$. This extension is only a notational convention for the revenue benchmark and the monopoly price; pricing policies themselves remain real-valued.

We assume that the valuation distribution $\F$ belongs to $\cF_\alpha$, the widely studied class of $\alpha$-regular distributions.
We restrict attention to distributions with interval support
$[\underline v_{\F},\overline v_{\F}]\subseteq\R_+,$
where \(\overline v_{\F}\) may be infinite. Each distribution is absolutely continuous, with density \(f>0\) on the interior of its support, and may additionally have an atom at \(\overline v_{\F}\) when \(\overline v_{\F}<\infty\); no other atoms are allowed. Accordingly, \(f\) denotes the density on the continuous part of the support. 
\begin{definition}
The set of $\alpha$-regular distributions is defined as follows:
\begin{equation}
    \cF_\alpha := \left\{\F \mid \varphi_{\alpha,\F}(v):=(1-\alpha)v - \frac{\barf(v)}{f(v)} \text{ is non-decreasing on } (\underline v_\F,\overline v_\F) \right\}~.
    \label{def:alpha regular set}
\end{equation}
We refer to $\varphi_{\alpha,\F}$ as the $\alpha$-virtual value function of the distribution $\F$, defined on the continuous interior of its support.
\end{definition}
When $\alpha=1$, the set $\cF_1$ represents the class of MHR distributions, and when $\alpha=0$, the set $\cF_0$ represents the class of regular distributions. 
For \(\alpha=1\), the density-based hazard rate \(f(v)/\barf(v)\) is defined only on the continuous interior; \(f\) is never evaluated at a finite upper-endpoint atom.
Note that \(p^*(\F)<+\infty\) whenever
\(\alpha>0\). Indeed,
\(\varphi_{0,\F}(v)=\varphi_{\alpha,\F}(v)+\alpha v\) is eventually
positive because \(\varphi_{\alpha,\F}\) is non-decreasing, while
\(\Rev'(v,\F)=-f(v)\varphi_{0,\F}(v)\) on the continuous interior.
Hence revenue is eventually decreasing when $v$ is large enough. When \(\alpha=0\),
\(p^*(\F)=+\infty\) may occur, as allowed by our extended-price
convention.

Next, we introduce an important auxiliary function that is tied to the definition of $\alpha$-regular distributions.

\begin{definition}
For every $\alpha \in [0,1]$, define the function $\cE_\alpha : \R \to [0,1]$ by
\[
\forall v, \quad \cE_\alpha(v) :=
\begin{cases}
\of{1+(1-\alpha)v}^{-1/(1-\alpha)}, & \text{if } \alpha\in[0,1)~;\\
e^{-v}, & \text{if } \alpha = 1~.
\end{cases}
\]
\end{definition}
The function $\cE_\alpha$ is the survival function of a generalized Pareto distribution, which corresponds to the extreme $\alpha$-regular distribution in the sense that its $\alpha$-virtual value is constant: for the distribution $\F$ with survival function $\cE_\alpha$ (i.e., $\barf(v) = \cE_\alpha(v)$, $\forall v$), we have $\varphi_{\alpha,\F}(v)=-1$ for all $v$. 

We summarize two standard properties of $\alpha$-regular distributions that will be used throughout our analysis.
The first lemma provides an alternative characterization of $\alpha$-regularity.
The second lemma establishes a lower bound on the optimal selling probability for $\alpha$-regular distributions.

\begin{lemma}[\cite{schweizer2015quantitative}]
A distribution $\F$ is $\alpha$-regular if and only if the function $v \to \cE_\alpha^{-1}(\barf(v))$ is convex on the interior of its support.
    \label{lem:cE_convex}
\end{lemma}
At a finite upper endpoint, this transform is evaluated using the weak-survival value \(\barf(v)=\Pr[V\sim\F]{V\ge v}\).

\begin{lemma}[\cite{schweizer2015quantitative}]
\label{lem:selling_prob}
For every $\alpha \in (0,1]$ and every $\F \in \cF_\alpha$,
\[
\barf\bigl(p^*(\F)\bigr) \ge \alpha^{1/(1-\alpha)}~.
\]
When $\alpha = 1$, the expression on the right-hand side is interpreted in the limiting sense as $1/e$.
\end{lemma}

Henceforth, we fix \(\alpha\in[0,1]\) and suppress the subscript \(\alpha\) to simplify notation.
We next define the following minimal and maximal $\alpha$-regular distributions, which are central to our analysis:
\begin{equation}
\sF(v) := 
\begin{cases}  1-\cE(\lambda v) &  v\in [0,a) \\
1 &  v \in [a,\infty)  
\end{cases} 
\quad \text{and} \quad
\lF(v) := 
\begin{cases}
0 &  v\in [0,b) \\
    1-\cE(\lambda(v-b)) ,\,  &  v\in [b,\infty)
\end{cases}
~.
\label{eq:worstF}   
\end{equation}
Observe that $\sF$ is obtained by scaling the generalized Pareto distribution defined above, by a factor of $1/\lambda$ and then truncating it at $a$. Similarly, $\lF$ is obtained by scaling the generalized Pareto random variable by $1/\lambda$ and then shifting it by $b$. In particular, the truncation creates an atom at \(a\): one has \(\sF(a)=1\), \(\barsF(a)=\cE(\lambda a)>0\), and thus \(\Rev(a,\sF)=a\cE(\lambda a)\).
We calculate the optimal posted prices under distributions $\sF$ and $\lF$.

\begin{lemma}
\label{lem:opt_price}
The optimal prices for $\sF$ and $\lF$ are given as follows:
\begin{itemize}
\item For $\sF$: if $\alpha=0$ or $\alpha \lambda a \le 1$, then $p^*(\sF) = a$; otherwise $p^*(\sF) = \frac{1}{\alpha \lambda}$~; %
\item For $\lF$: if $\lambda b \ge 1$, then $p^*(\lF) = b$; if \(\lambda b<1\) and \(\alpha>0\), then \(p^*(\lF)=\frac{1-(1-\alpha)\lambda b}{\alpha\lambda}\). If \(\alpha=0\), \(\lambda>0\), and \(\lambda b<1\), then \(p^*(\lF)=+\infty\) and \(\Rev(p^*(\lF),\lF)=\lim\limits_{x\to+\infty}\Rev(x,\lF)=1/\lambda=\opt(\lF)\).
\end{itemize}
\end{lemma}

Motivated by the preceding lemma, let $\scF$ and $\lcF$ denote the classes of minimal and maximal $\alpha$-regular distributions whose optimal posted prices are $a$ and $b$, respectively:
\begin{align}
\scF := \left\{ \sF \mid a>0,\ \lambda \ge 0, \, \alpha \lambda a \le 1 \right\} \quad \text{and} \quad 
\lcF := \left\{ \lF \mid b,\lambda \ge 0, \, \lambda b\ge 1 \right\}~.
\label{eq:worstF-set}
\end{align}

Recall that for every regular distribution $\F$, the revenue function $\Rev(v,\F)$ is unimodal. Accordingly, we define the critical price interval to be the set of prices that are $\Gamma$-approximate.

\begin{definition}[Critical Price Bounds]
\label{def: LBUB}
For every \(\F\in\cF\) and \(\Gamma\in[0,1]\), define
\begin{align*}
\mathsf{LB}(\F,\Gamma)
&:=\min\{z\in\overline{\R}_+:\Rev(z,\F)\ge\Gamma \cdot \opt(\F)\},\\
\mathsf{UB}(\F,\Gamma)
&:=\max\{z\in\overline{\R}_+:\Rev(z,\F)\ge\Gamma \cdot \opt(\F)\}.
\end{align*}
This superlevel set is nonempty because it contains \(p^*(\F)\). Extended revenue is upper semicontinuous and unimodal on \(\overline{\R}_+\), so the set is a closed interval whose endpoints are \(\mathsf{LB}(\F,\Gamma)\) and \(\mathsf{UB}(\F,\Gamma)\); either endpoint may equal \(+\infty\).
\end{definition}
\begin{lemma}
\label{lem:lb-ub}
For every \(\F\in\cF\), \(\Gamma\in[0,1]\), and \(z\in\overline{\R}_+\),
\[
\Rev(z,\F) \ge \Gamma \cdot \opt(\F) \iff z \in [\mathsf{LB}(\F,\Gamma), \mathsf{UB}(\F,\Gamma)]~.
\]
\end{lemma}

Finally, we prove that $\mathsf{LB}$ is Lipschitz continuous in $a$ over distributions in $\scF$, for any fixed $\Gamma$.
\begin{lemma}
\label{lem:lb-lipschitz}
For every \(\sF,\sF[a']\in\scF\) and \(\Gamma\in[0,1]\), \(|\mathsf{LB}(\sF,\Gamma)-\mathsf{LB}(\sF[a'],\Gamma)|\le|a-a'|\).
\end{lemma}
Building on the definition and structural properties of the minimal and maximal $\alpha$-regular distributions, the next subsection develops a reduction of nature’s inner minimization problem---originally infinite-dimensional and intractable---into a substantially simpler optimization problem.

\subsection{\texorpdfstring{Worst-Case Reduction}{Worst-Case Reduction}}

Evaluating a pricing policy \(p\) requires minimizing its performance ratio over \(\cF\):
\[
\inf_{\F \in \cF} \frac{\Rev(p,\F)}{\OPT(\F)}~.
\]
This is an infinite-dimensional, nonconvex optimization problem. We now impose the monotonicity condition used by our reduction. Write \(\F_1\distge\F_2\) when \(\F_1\) first-order stochastically dominates \(\F_2\), meaning that \(\F_1(v)\le\F_2(v)\) for every \(v\in\R\).

\begin{definition}[Monotone Pricing Policies]
A functional \(p:\Delta(\R_+)\to\R\) is monotone if \(p(\F_1)\ge p(\F_2)\) whenever \(\F_1\distge\F_2\).
\label{def: monotone}
\end{definition}

Concavity characterizes implementability by sample-contingent mechanisms, whereas monotonicity is the structural condition used in the worst-case analysis. Our main reduction shows that, for every monotone pricing policy, the infimum over all \(\alpha\)-regular distributions equals the infimum over the two-parameter boundary family \(\scF\cup\lcF\) defined in \eqref{eq:worstF-set}. The result does not require the policy to be concave.

\begin{theorem}
For any monotone pricing policy $p:\Delta(\R_+)\to\R$,
    \[
    \inf_{\F\in \cF} \frac{\Rev(p,\F)}{\OPT(\F)}=\inf_{\F\in \scF \cup \lcF} \frac{\Rev(p,\F)}{\OPT(\F)}~.
    \]
\label{thm:monotone}
\end{theorem}
The proof has two parts. If \(p(\F)\le p^*(\F)\), Lemma~\ref{lem:small} produces, for every tolerance, a lower boundary distribution in \(\scF\) whose performance ratio exceeds that under \(\F\) by no more than the tolerance. If \(p(\F)>p^*(\F)\), Lemma~\ref{lem:large} gives the analogous conclusion using \(\lcF\). The extended convention assigns \(p^*(\F)=+\infty\) to the first case and therefore gives the same case split for every \(\F\in\cF\).

\begin{lemma}
    Let \(p:\Delta(\R_+)\to\R\) be any monotone pricing policy, and suppose that \(p(\F)\le p^*(\F)\). For every \(\varepsilon>0\), there exists \(\F_\varepsilon^-\in\scF\), with \(p(\F_\varepsilon^-)\le p^*(\F_\varepsilon^-)\), such that
    \[
    \frac{\Rev(p(\F_\varepsilon^-),\F_\varepsilon^-)}{\OPT(\F_\varepsilon^-)}
    \le \frac{\Rev(p(\F),\F)}{\OPT(\F)}+\varepsilon.
    \]
    \label{lem:small}
\end{lemma}
\begin{proof}[Proof sketch]
Suppose first that \(v_0=p^*(\F)<+\infty\), and let
\(g(v)=\cE^{-1}(\bar F(v))\). By \(\alpha\)-regularity, \(g\) is
convex. We construct the lower boundary distribution \(\sF\) by
replacing \(g\) on \([0,v_0]\) with the secant line joining
\((0,0)\) and \((v_0,g(v_0))\), mapping this line back through
\(\cE\), and truncating at \(v_0\); see \Cref{fig:largeq}.
Equivalently,
\[
\barsF(v)
=
\cE\!\left(\frac{v}{v_0}g(v_0)\right),
\qquad 0\le v<v_0.
\]
Since \(g\) is convex and \(\cE\) is decreasing,
\(\sF\distle\F\), and the two distributions have the same weak
survival probability at \(v_0\). Hence
\[
\Rev(v,\sF)\le \Rev(v,\F)
\quad\text{for all }v,
\qquad
\Rev(v_0,\sF)=\Rev(v_0,\F)=\OPT(\F).
\]
Moreover, \(v_0\) is a monopoly price of \(\sF\), so
\(\OPT(\sF)=\OPT(\F)\).
Monotonicity of the pricing policy gives
\(p(\sF)\le p(\F)\le v_0\). Since the revenue of \(\F\) is
non-decreasing up to \(v_0\),
\[
\Rev(p(\sF),\sF)
\le
\Rev(p(\sF),\F)
\le
\Rev(p(\F),\F).
\]
Together with \(\OPT(\sF)=\OPT(\F)\), this gives the desired
ratio comparison $\frac{\Rev(p(\sF),\sF)}{\OPT(\sF)} \le \frac{\Rev(p(\F),\F)}{\OPT(\F)}$.

If \(p^*(\F)=+\infty\), we apply the same construction at a finite
anchor \(a_\varepsilon\ge p(\F)\) for which
\(\Rev(a_\varepsilon,\F)\) is arbitrarily close to \(\OPT(\F)\).
The full proof is in Appendix~\ref{app:reduction-proofs}.
 \end{proof}

\begin{figure}[htbp]
\centering
\begin{subfigure}[t]{0.45\textwidth}
\includegraphics[width=\textwidth,keepaspectratio]{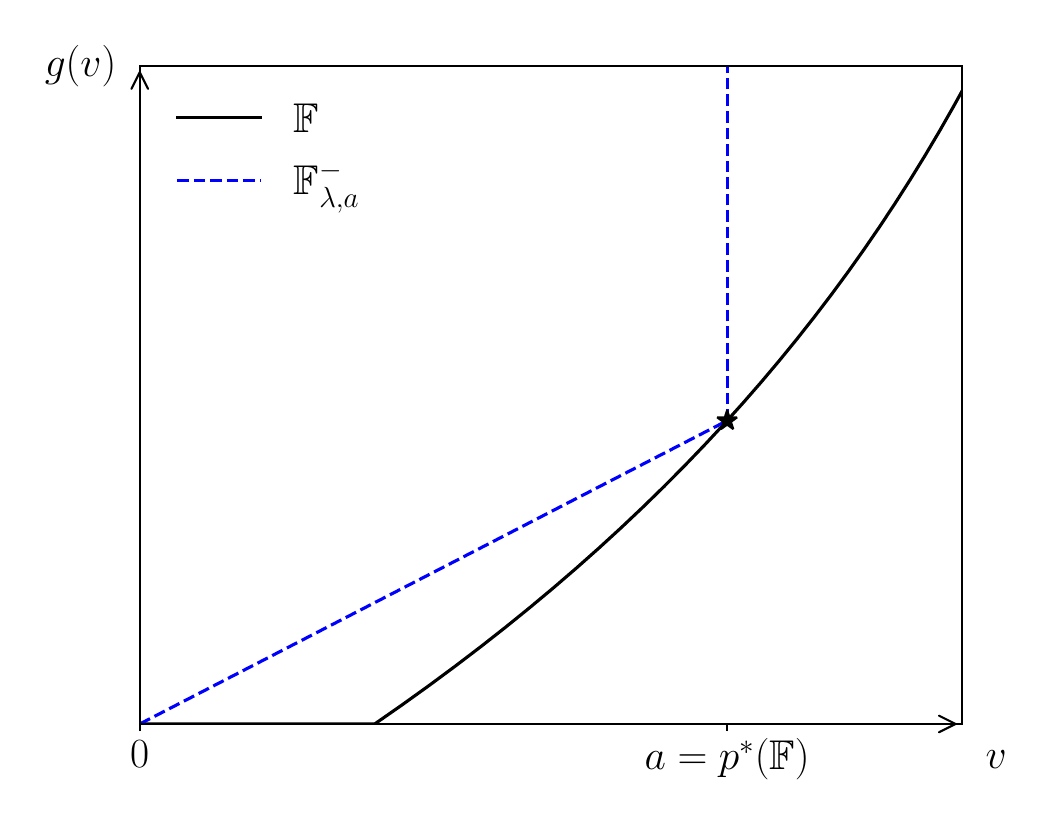}
\vspace{-25pt}
\caption{$g(v)$ for MHR distribution}
\label{fig:fig_regular_largeq}
\end{subfigure}
\begin{subfigure}[t]{0.45\textwidth}
\includegraphics[width=\textwidth,keepaspectratio]{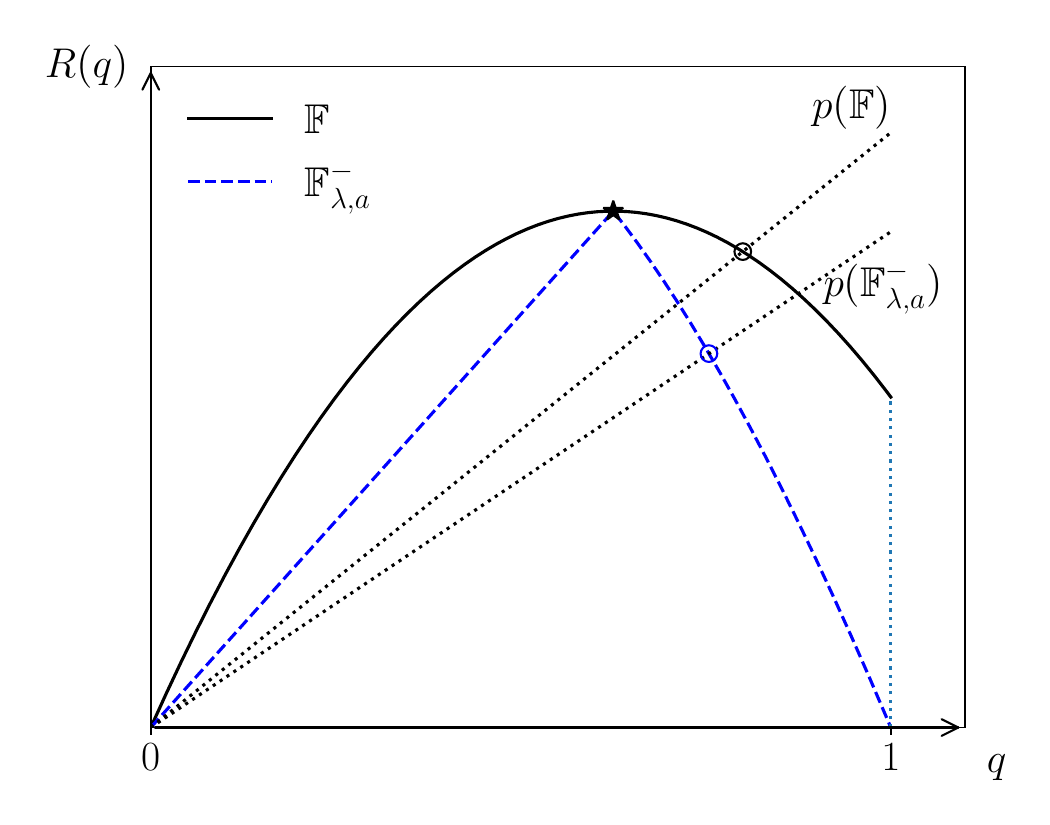}
\vspace{-25pt}
\caption{Revenue curve  for MHR distribution}
\label{fig:fig_mhr_largeq}
\end{subfigure}
\caption{\textbf{Lower-boundary construction in \Cref{lem:small}.}
(a) In the transformed space $g(v)=\cE^{-1}(\barf(v))$, the lower-boundary distribution $\sF$ is obtained by replacing the convex function $g$ on $[0,a]$, where $a=p^*(\F)$, with the secant joining $(0,0)$ and $(a,g(a))$, and truncating the distribution at $a$. Since the secant lies above $g$ and $\cE$ is decreasing, $\sF\preceq_1\F$, while the two distributions have the same survival probability at $a$.
(b) The corresponding revenue curves $R_{\F}(q)=q\,\mathrm{VaR}_q(\F)$. The construction preserves the optimal revenue, while monotonicity of the pricing policy moves the price weakly downward from $p(\F)$ to $p(\sF)$.}
\label{fig:largeq}
\end{figure}

\begin{lemma}
    Let \(p:\Delta(\R_+)\to\R\) be any monotone pricing policy, and suppose that \(p(\F)>p^*(\F)\). For every \(\varepsilon>0\), there exists \(\F_\varepsilon^+\in\lcF\), with \(p(\F_\varepsilon^+)\ge p^*(\F_\varepsilon^+)\), such that
    \[
    \frac{\Rev(p(\F_\varepsilon^+),\F_\varepsilon^+)}{\OPT(\F_\varepsilon^+)}
    \le
    \frac{\Rev(p(\F),\F)}{\OPT(\F)}+\varepsilon.
    \]
    If \(\barf(p(\F))>0\) and \(g(v):=\cE^{-1}(\barf(v))\) has a finite supporting slope at \(p(\F)\), the conclusion holds with \(\varepsilon=0\).
    \label{lem:large}
\end{lemma}

\begin{proof}[Proof sketch]
Let \(v_0=p(\F)>p^*(\F)\) and \(g(v)=\cE^{-1}(\barf(v))\).
In the generic case, take a supporting line of the convex function \(g\) at \(v_0\) and map it back through \(\cE\) to construct an upper boundary distribution \(\lF\); see \Cref{fig:smallq}. Since the supporting line lies below \(g\) and \(\cE\) is decreasing, \(\lF\distge\F\), while \(\lF\) and \(\F\) have the same weak survival probability at \(v_0\). Hence
\[
\Rev(v_0,\lF)=\Rev(v_0,\F),
\qquad
\OPT(\lF)\ge\OPT(\F).
\]
The supporting slope is chosen so that \(p^*(\lF)\le v_0\). Monotonicity then gives
\(p(\lF)\ge p(\F)=v_0\). Since the revenue of \(\lF\) is non-increasing above its monopoly price,
\[
\Rev(p(\lF),\lF)
\le
\Rev(v_0,\lF)
=
\Rev(p(\F),\F),
\]
which, together with \(\OPT(\lF)\ge\OPT(\F)\), gives the desired ratio comparison $\frac{\Rev(p(\lF),\lF)}{\OPT(\lF)} \le \frac{\Rev(p(\F),\F)}{\OPT(\F)}$.

If the resulting \(\lF\notin\lcF\), we normalize it at its monopoly price to obtain \(\llF\in\lcF\). The tail above the monopoly price is preserved up to a common multiplicative factor, while \(\llF\distge\lF\); hence the normalization cannot increase the performance ratio. The full proof is in Appendix \ref{app:reduction-proofs}.

\end{proof}

\begin{figure}[htbp]
\centering
\begin{subfigure}[t]{0.45\textwidth}
\includegraphics[width=\textwidth,keepaspectratio]{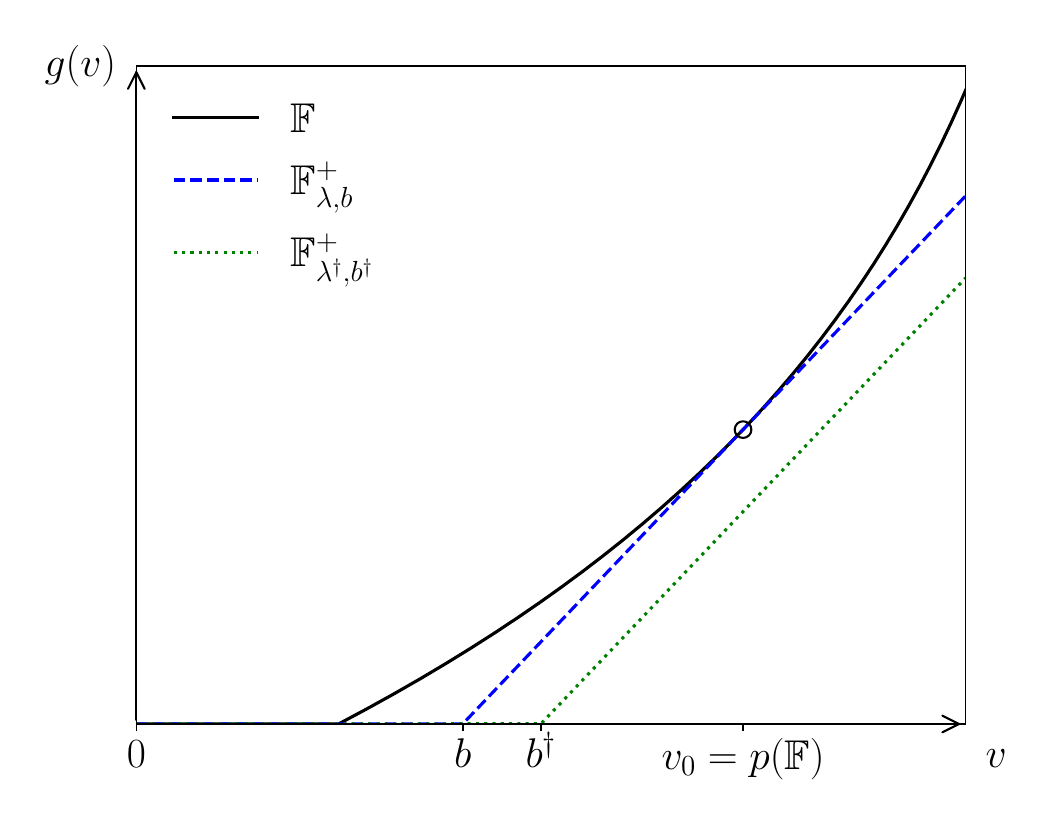}
\vspace{-25pt}
\caption{$g(v)$ for MHR distribution}
\label{fig:fig_regular_smallq}
\end{subfigure}
\begin{subfigure}[t]{0.45\textwidth}
\includegraphics[width=\textwidth,keepaspectratio]{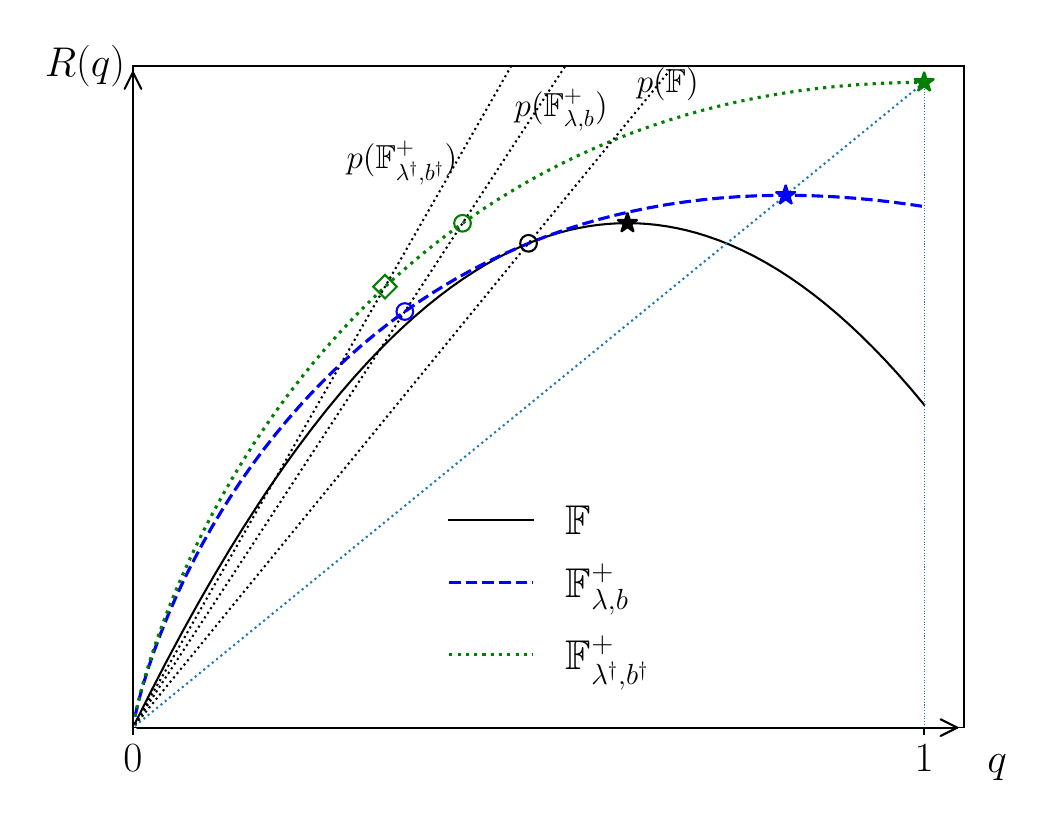}
\vspace{-25pt}
\caption{Revenue curve for MHR distribution}
\label{fig:fig_mhr_smallq}
\end{subfigure}
\caption{\textbf{Upper-boundary construction in \Cref{lem:large}}. (a) In the transformed space \(g(v):=\cE^{-1}(\barf(v))\), \(\lF\) is generated by a supporting line to \(g\) at \(v_0=p(\F)\); when necessary, it is normalized to \(\llF\).
(b) The corresponding revenue curves, where monotonicity moves the policy price weakly upward, and normalization weakly decrease the performance ratio.}
\label{fig:smallq}
\end{figure}

\begin{proof}[Proof of \Cref{thm:monotone}]
Because \(\scF\cup\lcF\subseteq\cF\), the infimum over \(\cF\) is at most the infimum over \(\scF\cup\lcF\). For the reverse inequality, fix \(\F\in\cF\). Apply \Cref{lem:small} when \(p(\F)\le p^*(\F)\) and \Cref{lem:large} when \(p(\F)>p^*(\F)\); the first case includes the extended optimizer \(p^*(\F)=+\infty\). For every \(\varepsilon>0\), the applicable lemma gives a boundary distribution whose performance ratio is no greater than the ratio under \(\F\) plus \(\varepsilon\). Letting \(\varepsilon\downarrow0\), we obtain
\[
\inf_{\F'\in\scF\cup\lcF}
\frac{\Rev(p,\F')}{\OPT(\F')}
\le
\frac{\Rev(p,\F)}{\OPT(\F)}
\qquad\text{for every }\F\in\cF.
\]
Taking the infimum over \(\F\in\cF\) proves the reverse inequality and the theorem.

\end{proof}

Finally, we state an easy-to-verify corollary of Theorem~\ref{thm:monotone}, which we will later use to identify the approximation ratios attained by different monotone pricing policies.

\begin{corollary}
\label{cor:monotone}
    For \(\Gamma\in[0,1]\), a monotone pricing policy \(p:\Delta(\R_+)\to\R\) is \(\Gamma\)-approximate if and only if
    \begin{itemize}
\item for every $\sF \in \scF$, $p(\sF) \ge \lb$;
\item for every $\lF \in \lcF$, $p(\lF) \le \ub$.
\end{itemize}
where $\lb$ and $\ub$ are defined in \Cref{def: LBUB}.
\end{corollary}

\Cref{thm:monotone} reduces the infinite-dimensional worst-case problem to a low-dimensional optimization over boundary distributions. It therefore makes both theoretical analysis and numerical computation tractable and provides a common framework for all monotone pricing policies.
In particular, by instantiating the functional $p(\F)$ as the mean $\|\F\|_1$, the $L^\eta$-norm $ \|\F\|_\eta$ or the superquantile $\mathrm{CVaR}_{q}(\F)$, we can characterize the optimal mean pricing, $L^\eta$-norm pricing, and superquantile pricing in the following problem, respectively. In each case, the optimal policy is obtained by choosing the optimal discount factor $\omega$.
\begin{align}
    \max_{\omega} \min_{\F\in \cF} \frac{\Rev\of{\omega \cdot p(\F), \F}}{\OPT(\F)}
    \label{eq: linear}
\end{align}
Solving \eqref{eq: linear} yields the optimal policy within each class of statistic-based policies and its approximation ratio; the proofs coincide with those in Section~\ref{sec:application}.
\begin{proposition}
    For MHR distributions, 
    \begin{enumerate}
        \item Mean pricing $p(\F)= 0.823 \cdot \|\F\|_1$ achieves an approximation ratio of 0.771.
        \item $L^{1.37}$-norm pricing $p(\F)= 0.805 \cdot \|\F\|_{1.37}$ achieves an approximation ratio of 0.774. 
        \item Superquantile pricing $p(\F) = 0.787 \cdot \mathrm{CVaR}_{0.92}(\F)$ achieves an approximation ratio of 0.787.
    \end{enumerate}
\end{proposition}

Our reduction in \Cref{thm:monotone} not only enables the evaluation of specific pricing policies, but also allows us to characterize the optimal monotone sample-contingent mechanism (equivalently, the optimal monotone concave pricing policy) and its performance, which is presented in \Cref{sec:optimal}.
\subsection{Technical Comparison with Existing Analysis}

Discounted mean pricing, \(p(\F)=\omega\,\E_{\F}[v]\), appears as a special case in two different settings: the sample-bid mechanism of \cite{stoc/FengHL21} for one-sample prior-independent mechanism design, and the convex-surrogate approach of \cite{ms/biggs2026} for contextual pricing with transaction data. Both works separate the underpricing and overpricing regimes according to whether $p(\F)$ lies below or above the monopoly price. Both analysis leverage the relationship between the mean and the tail geometry of MHR distributions: \cite{stoc/FengHL21} relates the mean and the monopoly quantile through MHR tail inequalities, and \cite{ms/biggs2026} combines the survival-function representation of the mean with log-concavity to bound revenue in the two pricing regimes.

Our analysis shows that neither feature is essential. \Cref{thm:monotone} identifies a more general principle behind this two-regime analysis. The reduction uses only that the pricing functional is monotone with respect to first-order stochastic dominance, while the distributional argument applies to the entire family of \(\alpha\)-regular distributions. The same pair of boundary constructions therefore covers mean pricing, quantile pricing, and every other monotone pricing policy. In this sense, \Cref{thm:monotone} places the previously separate mean- and quantile-based analysis within a common framework and extends them from MHR distributions to general \(\alpha\)-regular distributions.

\section{Optimal Monotone Concave Pricing}
\label{sec:optimal}

\Cref{thm:monotone} makes it straightforward to optimize a scalar parameter within a specified family, such as discounted mean, $L^\eta$-norm, or superquantile pricing. Optimizing over all monotone sample-contingent mechanisms is substantially harder because the buyer may report an entire distribution and the seller must specify a transfer schedule for every possible report. In this section, we characterize the optimal monotone sample-contingent mechanism over $\cF_\alpha$ for every $\alpha\in[0,1]$. Equivalently, we characterize the optimal monotone concave pricing policy and then construct a transfer rule that implements it.

By Theorem~\ref{thm:monotone}, once monotonicity of $p$ is imposed, it suffices to verify the approximation ratio for distributions in $\scF \cup \lcF$. 
Pricing policies remain real-valued, so the extended critical-price intervals below are evaluated at finite policy prices.
Consider the following construction. We start by setting $p(\sF)=\lb $ for all $\sF\in\scF$, and then extend $p$ to all distributions so that it is both concave and monotone. Equivalently, $p$ is defined as the concave and monotone envelope of $\mathsf{LB}(\cdot,\Gamma)\vert_{\scF}$. This construction guarantees that $p$ satisfies both concavity and monotonicity. Moreover, this construction yields the minimal concave and monotone functional $p$ that satisfies the first condition of Corollary~\ref{cor:monotone}.
Therefore, if this extension additionally satisfies $p(\lF)\le \ub$ for all $\lF \in \lcF$, we may conclude that $p$ achieves the desired approximation guarantee. This is indeed our construction, and $\Gamma$ is chosen to be the largest constant for which the above condition holds.

Although Theorem~\ref{thm:scoring} characterizes how the corresponding transfer rule $h$ can be derived from $p$, we provide a more intuitive construction below. Moreover, this explicit construction enables us to compute the optimal approximation ratio numerically.

We first restrict our transfer rule $h$ to the following family of feasible transfer schedules.

\begin{definition}[Feasible Transfer Schedules]
For every \(\Gamma\in[0,1]\), define \(H^\Gamma\) to be the following space of feasible transfer schedules:
\[
H^\Gamma := 
\left\{
h: \R \to \R \;\middle|\;
\begin{aligned}
& \E_{s \sim \sF}\off{h(s)} \ge \lb, \quad \forall (\lambda, a) \in K := \{(\lambda,a)\in\R_+^2 \mid \alpha \lambda a \le 1\} \\
& h \text{ is non-decreasing}
\end{aligned}
\right\}~.
\]
\end{definition}
The \(a=0\) member of \(K\) is retained only as a closure point; its constraint is vacuous because the corresponding lower revenue bound is zero. It is not a distribution in \(\scF\) and does not enter any performance ratio.
Intuitively, the seller delegates the selection of the transfer schedule to the buyer, but restricts her choice to the admissible family $H^\Gamma$. Since the buyer knows the underlying $\F$, she will choose the schedule $h\in H^\Gamma$ that minimizes the induced interim price. This mirrors the intuition seen in earlier examples such as the $L^2$-norm or superquantile transfer rules, where the buyer’s best response identifies the statistic of interest. 
Formally, we define a proper transfer rule as follows.

\begin{definition}[Optimal Monotone Transfer Rule]
For every \(\Gamma\in[0,1]\), define \(h^\Gamma:\R_+\times\Delta(\R_+)\to\R\) to be the following transfer rule:
\[
\forall \F \in \Delta(\R_+), \quad h^\Gamma(\cdot,\F) = \argmin_{h \in H^\Gamma} \Ex{s \sim \F}{h(s)}~.
\]
\end{definition}

For a given parameter $\Gamma$, the seller commits ex ante to the transfer rule $h^\Gamma$.  
Then a buyer who knows $\F$ will truthfully report $\F$ since $h^\Gamma(\cdot, \F)$ minimizes her expected payment among all feasible transfer schedules in $H^\Gamma$. By the definition of $H^\Gamma$, a larger $\Gamma$ imposes a larger $\lb$ across all distributions $\sF$ in $\scF$, which shrinks the feasible set $H^\Gamma$ and results in a higher expected payment $\Ex{s \sim \F}{h^\Gamma(s,\F)}$. Intuitively, increasing $\Gamma$ leads to ``larger'' $h^\Gamma$. On the other hand, higher prices can worsen performance against distributions $\lF$ in $\lcF$, where overpricing leads to a lower approximation ratio. Consequently, the choice of $\Gamma$ involves a fundamental tradeoff: increasing $\Gamma$ improves guarantees against distributions in 
$\scF$ but may degrade performance against distributions in 
$\lcF$. The optimal transfer rule must balance the performance across these two classes.
We therefore define the largest value of $\Gamma$ for which the resulting transfer rule achieves the same approximation ratio against all distributions in $\lcF$ and prove that it is the tight approximation ratio.
\begin{definition}[Tight Approximation Ratio]
For every $\alpha \in [0,1]$, define
\[
\Gamma^* := 
\max_{\Gamma \in [0,1]}
\left\{
\Gamma \;\middle|\;
\E_{s \sim \lF}\!\left[h^{\Gamma}(s,\lF)\right] \le \ub, 
\quad \forall\, \lF \in \lcF
\right\}.
\]
\end{definition}
\Cref{thm:optimal} establishes the optimal monotone sample-contingent mechanism and the optimality of $\Gamma^*$.
\begin{theorem}
\label{thm:optimal}
Over the family of $\alpha$-regular distributions $\cF_\alpha$, the sample-contingent mechanism induced by $h^{\Gamma^*}$ attains a $\Gamma^*$-approximation ratio;
moreover, no monotone sample-contingent mechanism can achieve an approximation ratio strictly greater than $\Gamma^*$.
\end{theorem}

The proof of the theorem is divided into two parts. In Lemma~\ref{lem:positive}, we provide a sufficient condition to prove that $h^\Gamma$ is $\Gamma$-approximate. 
In Lemma~\ref{lem:negative}, we prove the same condition is also necessary for the existence of a $\Gamma$-approximation monotone sample-contingent mechanism.

By the definition of $H^\Gamma$, the first constraint of $H^\Gamma$ enforces the starting point of our construction (i.e., $p(\sF)=\lb$), while the second constraint imposes the monotonicity of $p$. We emphasize that it is far from obvious that the transfer schedule $h$ must be monotone for every distribution $\F$. It is straightforward to verify that this condition is sufficient to ensure the monotonicity of $p$, which makes the proof of Lemma~\ref{lem:positive} relatively direct. 
In contrast, establishing its necessity is significantly more subtle. This issue forms the core of our analysis in Lemma~\ref{lem:negative} and constitutes the most technically involved part of the proof.

Our main theorem then follows as a corollary of the two lemmas and the definition of $\Gamma^*$. 

\begin{lemma}
\label{lem:positive}
The sample-contingent mechanism with $h^{\Gamma}$ is $\Gamma$-approximate, if 
\[
\Ex{s \sim \lF}{h^{\Gamma}(s,\lF)} \le \ub, \quad \forall \lF \in \lcF~.
\]
\end{lemma}
\begin{proof}
    We first verify that the transfer rule $h^\Gamma$ is proper:
\[
\forall \F,\F', \quad \Ex{s \sim \F}{h^\Gamma(s,\F)} = \min_{h \in H^\Gamma} \Ex{s \sim \F}{h(s)} \le \Ex{s \sim \F}{h^\Gamma(s,\F')}~,
\]
where the inequality holds since $h^\Gamma(\cdot,\F') \in H^\Gamma$.

Next, we verify that $h^\Gamma$ is monotone:
\[
\forall \F \distge \F', \quad \Ex{s \sim \F}{h^\Gamma(s,\F)} \ge \Ex{s \sim \F'}{h^\Gamma(s,\F)} \ge \min_{h \in H^\Gamma} \Ex{s \sim \F'}{h(s)} = \Ex{s \sim \F'}{h^\Gamma(s,\F')}~,
\]
where the first inequality holds by a standard coupling argument and the fact that $h^\Gamma(\cdot,\F)$ is non-decreasing.
Hence, by Corollary~\ref{cor:monotone}, the sample-contingent mechanism induced by $h^\Gamma$ is $\Gamma$-approximate if and only if:
\begin{itemize}
    \item $\forall \sF \in \scF$, $\E_{s \sim \sF}[h^\Gamma(s,\sF)] \ge \lb$~;
    \item $\forall \lF \in \lcF$, $\E_{s \sim \lF}[h^\Gamma(s,\lF)] \le \ub$~.
\end{itemize}
Notice that the first family of conditions holds by the definition of $H^\Gamma$. Indeed, we have $h^\Gamma(\cdot,\sF) \in H^\Gamma$, and every $h \in H^\Gamma$ satisfies that
$\Ex{s \sim \sF}{h(s)} \ge \lb, \, \forall \sF \in \scF~.$
Observe that the second family of conditions is the stated assumption of the lemma. This concludes the proof. 

\end{proof}

\begin{lemma}
\label{lem:negative}
If there exists a $\Gamma$-approximate monotone sample-contingent mechanism, then we have 
\[
\Ex{s \sim \lF}{h^{\Gamma}(s,\lF)} \le \ub, \quad \forall \lF \in \lcF~.
\]
\end{lemma}
\begin{proof}
For an arbitrary distribution $\F\in\Delta(\R_+)$, consider the following optimization. We abuse P($\F$) to denote both the program and its value. 
\begin{align}
\label{lp:primal}
\max_{\G \in \Delta(K)}: \quad  &  \Ex{(\lambda,a) \sim \G}{\lb} = \int_K \lb \dd \G \tag{P($\F$)}\\
\text{subject to}: \quad & \int_{K} \barsF(s)  \dd \G \le \barf(s) & \forall s \in \R \notag
\end{align}
Recall that $h^\Gamma(\cdot,\F)$ is the optimal solution of the following program:
\begin{align}
\label{lp:dual}
\min_{h:\R \to \R}: \quad  & \Ex{s \sim \F}{h(s)} \tag{D($\F$)}\\
\text{subject to}: \quad & \Ex{s \sim \sF}{h(s)} \ge \lb & \forall (\lambda, a) \in K \notag \\
    & h(\cdot) \text{ is non-decreasing} \notag
\end{align}

We verify that the two programs form a primal-dual pair.
Let the Lagrangian multiplier be a non-negative finite Borel measure $\mu \in \mathcal{M}(\R)$. The Lagrangian is
\begin{align*}
\mathcal{L}(\G,\mu) & = \int_{K} \lb \dd \G - \int_{\R} \int_{K} \barsF(s)  \dd \G \dd \mu + \int_{\R} \barf(s) \dd \mu \\
& = \int_{K} \left( \lb - \int_{\R} \barsF(s) \dd \mu \right) \dd \G + \int_{\R} \barf(s) \dd \mu~.
\end{align*}
Because $\G$ ranges over probability measures on $K$,
\[
\sup_{\G} \mathcal{L}(\G,\mu) = \int_{\R} \barf(s)\dd \mu + \sup_{(\lambda,a) \in K} \left( \lb - \int_{\R_+} \barsF(s) \dd \mu \right)~.
\]
Thus, the dual problem is
\begin{align*}
\min_{\mu}: \quad  & \int_{\R} \barf(s) \dd \mu \\
\text{subject to}: \quad & \int_{\R} \barsF(s) \dd \mu \ge \lb & \forall (\lambda,a) \in K
\end{align*}
Define the cumulative $h(s) = \mu([0,s])$. Through integration-by-parts, we have that 
\[
\int_{\R} \barf(s) \dd \mu = \int_{\R} h(s) \dd \F(s) \quad \text{and} \quad \int_{\R} \barsF(s) \dd \mu = \int_{\R} h(s) \dd \sF(s)~.
\] 
Moreover, $\mu \ge 0$ implies that $h$ is non-negative and non-decreasing. Therefore, the program is equivalent to \eqref{lp:dual}. 

As a consequence, we have weak duality that \ref{lp:primal} $\le$ \ref{lp:dual}. 
However, we cannot directly invoke Sion’s minimax theorem to claim strong duality, even though $\mathcal{L}(\G,\mu)$ is bilinear in $\G$ and $\mu$. The difficulty is that the domain \(K\) defined above is unbounded. Consequently, the feasible set of distributions supported on $K$ is not compact under the weak topology, and Sion's compactness requirement fails. 
In the next lemma, we prove that strong duality holds for every $\lF \in \lcF$ by introducing a carefully chosen truncation of $K$, which restores compactness and allows the minimax argument to go through.

\begin{lemma} \label{lem::primal-dual-used-discrete}
For every $\lF \in \lcF$, we have
\begin{equation}
\label{eqn:duality}
\mathrm{P}(\lF) = \mathrm{D}(\lF)~.
\end{equation}    
\end{lemma}
\begin{proof}
Fix an arbitrary $\bar{a} \ge b$, and define the truncated domain $\bar{K} = K \cap \{(\lambda,a) : a \in [b, \bar{a}] \}$. 
Since $\alpha \lambda a \le 1$ and $a \ge b$, the variable $\lambda$ is automatically bounded by $1/\alpha b$; hence $\bar{K}$ is compact.
Let $\bar{\mathrm{P}}(\lF)$ be the primal program obtained by restricting $\G$ to distributions supported on $\bar{K}$, and let $\bar{\mathrm{D}}(\lF)$ be the corresponding dual program.
Then, we have
\[
\mathrm{P}(\lF) \ge \bar{\mathrm{P}}(\lF) = \bar{\mathrm{D}}(\lF)~.
\]
The inequality follows because $\bar{K} \subseteq K$, so the feasible set of the truncated primal is a subset of that of (P($\lF$)). The equality follows from Sion’s minimax theorem: on the truncated domain $\bar{K}$, compactness is restored, and the Lagrangian remains convex–concave and bilinear in $\G$ and $\mu$, ensuring strong duality.

Let $\bar{h}$ be the optimal solution of ($\bar{\mathrm{D}}(\lF)$). Then $\bar{h}$ is non-decreasing and satisfies that 
\[
\Ex{x \sim \sF}{\bar{h}(x)} \ge \lb, \qquad \forall (\lambda, a) \in \bar{K}~.
\]
We extend $\bar{h}$ to a feasible solution of ($\mathrm{D}(\lF)$). Define $h$ as follows:\footnote{We restrict attention to $\alpha \in (0,1]$, as the expression $\alpha^{-\frac{1}{1-\alpha}}$ is ill-defined at $\alpha=0$. In Section~\ref{sec:negative}, we show that when $\alpha=0$, corresponding to the class of regular distributions, sample-contingent mechanisms—equivalently, concave pricing policies—cannot achieve any non-trivial (i.e., strictly positive) approximation ratio.}
\[
h(s) = \begin{cases} \bar{h}(b) & s \in [0,b] \\
\bar{h}(s) & s \in [b, \bar{a}] \\
\bar{h}(\bar{a}) + \alpha^{-\frac{1}{1-\alpha}} \cdot (s - \bar{a}) & s \in [\bar{a}, \infty) 
\end{cases}
\]
The monotonicity of $h$ follows straightforwardly by the definition. Next, we verify that
\[
\Ex{s \sim \sF}{h(s)} \ge \lb, \qquad \forall (\lambda, a) \in K~.
\]
\begin{itemize}
\item For $a \le b$, we have
\[
\Ex{s \sim \sF}{h(s)} = \bar{h}(b) = \Ex{x \sim \F^{-}_{0,b}}{\bar{h}(s)} \ge \mathsf{LB}(\F^{-}_{0,b},\Gamma) \ge \mathsf{LB}(\sF,\Gamma)~.
\]
Here, $\F^{-}_{0,b}$ is the point-mass distribution at $b$. The last inequality is due to $\F^{-}_{0,b} \distge \sF$.

\item For $a \in [b, \bar{a}]$, we have
\[
\Ex{s \sim \sF}{h(s)} \ge \Ex{s \sim \sF}{\bar{h}(s)} \ge \lb~,
\]
where the first inequality holds by the fact that $h(s) \ge \bar{h}(s)$, for all $s$.

\item For $a \ge \bar{a}$, we have
\begin{align*}
\Ex{s \sim \sF}{h(s)} & \ge \Ex{s \sim \sF[\bar{a}]}{\bar{h}(s)} + \alpha^{-\frac{1}{1-\alpha}} \cdot \int_{\bar{a}}^{a} (s-\bar{a}) \dd \sF(s) \\
& = \Ex{s \sim \sF[\bar{a}]}{\bar{h}(s)} + \alpha^{-\frac{1}{1-\alpha}} \cdot \int_{\bar{a}}^{a} \cE(\lambda s) \dd s \tag{integration by parts} \\
& \ge \mathsf{LB}(\sF[\bar{a}],\Gamma) + \alpha^{-\frac{1}{1-\alpha}} \cdot \cE(\lambda a) \cdot (a-\bar{a}) \tag{by the monotonicity of $\cE$}\\
& \ge \mathsf{LB}(\sF[\bar{a}],\Gamma) + \alpha^{-\frac{1}{1-\alpha}} \cdot \cE(\alpha^{-1}) \cdot (a-\bar{a}) \tag{by $\alpha \lambda a \le 1$}\\
& = \mathsf{LB}(\sF[\bar{a}],\Gamma) + (a-\bar{a}) \\
& \ge \lb~. \tag{by Lemma~\ref{lem:lb-lipschitz}}
\end{align*}
\end{itemize}
Therefore, $h \in H^\Gamma$ and we have
\begin{align*}
\mathrm{D}(\lF) \le \Ex{\lF}{h(s)} & = \Ex{\lF}{\bar{h}(s)} + \alpha^{-\frac{1}{1-\alpha}} \cdot \int_{\bar{a}}^{\infty} (s-\bar{a}) \dd \lF(s) \\
& = \bar{\mathrm{D}}(\lF) + \alpha^{-\frac{1}{1-\alpha}} \cdot \int_{\bar{a}}^{\infty} \cE(\lambda (s-b)) \dd s \\
& \le \mathrm{P}(\lF) + \alpha^{-\frac{1}{1-\alpha}} \cdot \int_{\bar{a}}^{\infty} \cE(\lambda (s-b)) \dd s~.
\end{align*}
As $\bar{a} \to \infty$, the tail integral tends to $0$, and so $\mathrm{D}(\lF) \le \mathrm{P}(\lF)$. Combined with weak duality, this yields $\mathrm{D}(\lF) = \mathrm{P}(\lF)$, completing the proof.

\end{proof}

Finally, we prove \Cref{lem:negative}.
Suppose there exists a transfer rule $h$ achieving an approximation ratio of $\Gamma$. Let $p_h(\F) = \Ex{s \sim \F}{h(s,\F)}$ denote the corresponding induced interim price, which is monotone and concave.
For any given $\lF \in \lcF$, let $\G(\lambda,a)$ be the optimal solution of the primal program ($\mathrm{P}(\lF)$). Consider the random mixture of distributions $\{\sF\}$, where $(\lambda,a)$ is drawn from $\G$; we denote the resulting distribution by $\tilde\F$. The feasibility constraints $\barf[\tilde\F](s) = \int_{K} \barsF(s) \dd \G \le \barlF(s)$ of the program ensure that the distribution $\tilde\F$ is stochastically dominated by $\lF$. Therefore,
\begin{align*}
\ub & \ge p_h(\lF) \tag{by Lemma~\ref{lem:lb-ub}} \\
& \ge p_h(\tilde\F) \ge \Ex{(\lambda, a)\sim \G}{p_h(\sF)} \tag{by the monotonicity and concavity of $p_h$} \\
& \ge \Ex{(\lambda, a)\sim \G}{\lb} \tag{by Lemma~\ref{lem:lb-ub}} \\
& = \mathrm{P}(\lF) \tag{by the definition of $\G$}\\
& = \Ex{s \sim \lF}{h^\Gamma(s,\lF)} \tag{by strong duality \eqref{eqn:duality}}~,
\end{align*}
which concludes the proof.

\end{proof}

Since the expressions for the lower and upper bounds $\mathsf{LB}$ and $\mathsf{UB}$ involve solutions to equations with Lambert W functions when $\alpha=1$, the approximation ratio $\Gamma^*$ is characterized implicitly. While these expressions are not elementary, they enable efficient numerical evaluation. In Appendix~\ref{sec:numerical}, we %
use computer-assisted experiments to compute $\Gamma^*$ for MHR distributions, showing that it lies between $0.79$ and $0.796$. The same approach extends to computing $\Gamma^*(\alpha)$ for any $\alpha \in [0,1]$.

\begin{proposition}
    Over the family of MHR distributions, the optimal monotone sample-contingent mechanism achieves an approximation ratio between $0.79$ and $0.796$.
    \label{prop: optimal numerical}
\end{proposition}

\section{Negative Results}
\label{sec:negative}
Monotonicity is natural and practically desirable, but it may restrict performance. We therefore derive upper bounds for arbitrary sample-contingent mechanisms, or equivalently, arbitrary concave pricing policies. For MHR distributions, these bounds show that imposing monotonicity entails only a small loss in the achievable approximation ratio.

\begin{theorem}
\label{thm:negative-concave}
For every $\alpha \in (0,1]$, define
\[
\overline{\Gamma}_{\alpha}
:=
\min_{b \ge 1}
\frac{
b+\cE_\alpha^{-1}\!\left(\frac{\alpha b}{1+\alpha b}\right)
}{
b+\frac{1}{\alpha}
}~.
\]
No concave pricing policy can achieve an approximation ratio strictly greater than $\overline{\Gamma}_{\alpha}$ over the family $\cF_\alpha$. In particular, no concave pricing policy can achieve an approximation ratio better than $0.801$ (resp. $0$) for MHR distributions (resp. regular distributions).
\end{theorem}
The resulting curve of $\overline{\Gamma}_{\alpha}$ is plotted in \Cref{fig:negative-concave-alpha}.
\begin{figure}[htbp]
\centering
\captionsetup{justification=centering}
\caption{Approximation ratio upper bound $\overline{\Gamma}_{\alpha}$ for concave pricing policies}
\label{fig:negative-concave-alpha}
\includegraphics[width=0.45\textwidth,keepaspectratio]{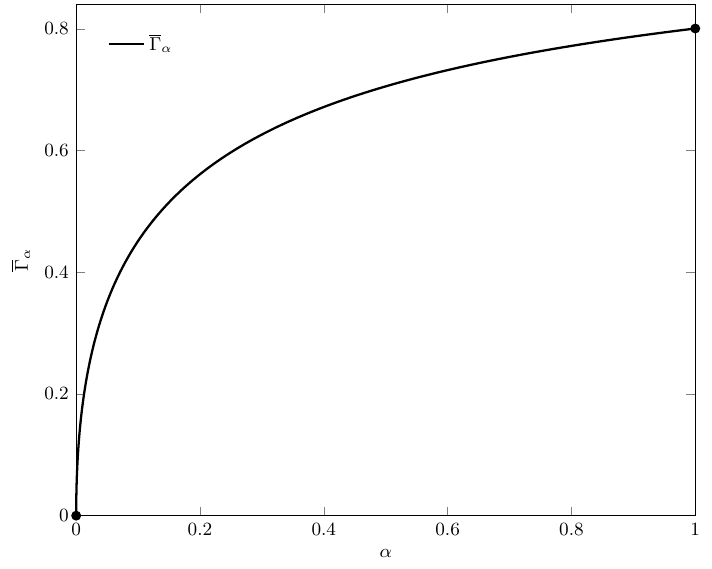}
\end{figure}
\begin{proof}
We first consider $\alpha \in (0,1]$. Suppose that a concave pricing policy $p:\Delta(\R_+)\to\R$ achieves a $\Gamma$ approximation ratio for some $\Gamma>0$. Since the optimal revenue under the point-mass distribution $\delta_v$ is $v$, the approximation requirement implies $p(\delta_v)\ge\Gamma v$ for every $v\in\R_+$. By the concavity of $p$, it follows that
\begin{equation}
p(\F)
\ge
\Ex{s\sim\F}{p(\delta_s)}
\ge
\Gamma\cdot\Ex{s\sim\F}{s}~.
\label{eq:concave-mean-lower-bound}
\end{equation}

Since we are proving an impossibility result, it suffices to test the pricing policy on any subfamily of $\alpha$-regular distributions. We restrict attention to $\lambda=1$. Fix an arbitrary $b\ge1$ and let $\F_b:=\F^{+}_{1,b}\in\lcF$. By \Cref{lem:opt_price}, $\opt(\F_b)=b$. Moreover,
\[
\Ex{s\sim\F_b}{s}
=
b+\int_0^\infty \cE_\alpha(u)\dd{u}
=
b+\frac{1}{\alpha}~.
\]
Combining \eqref{eq:concave-mean-lower-bound} with the requirement that $p(\F_b)$ be a $\Gamma$-approximate price gives the following necessary condition:
\begin{equation}
\mathsf{UB}(\F_b,\Gamma)
\ge
p(\F_b)
\ge
\Gamma \cdot \Ex{s\sim\F_b}{s}~.
\label{eq:concave-upper-bound-necessary}
\end{equation}

\begin{itemize}
    \item If $\Gamma \Ex{s\sim\F_b}{s}\le b$, then $\Gamma\le b/\Ex{s\sim\F_b}{s}$, which is already weakly smaller than the claimed upper bound.
    \item If $\Gamma \Ex{s\sim\F_b}{s}> b$, then $p(\F_b)
\ge
\Gamma \cdot \Ex{s\sim\F_b}{s}> b$. %
Since the revenue of $\F_b$ is decreasing to the right of its optimal price $b$, we have 
\end{itemize}
\[
\Gamma b = \Gamma \opt(\F_b)
\le
\Rev(p(\F_b),\F_b) 
\le  \Rev(\Gamma \cdot \Ex{s\sim\F_b}{s},\F_b)
=
\Gamma \cdot \Ex{s\sim\F_b}{s}\,
\cE_\alpha\!\left(\Gamma \cdot \Ex{s\sim\F_b}{s}-b\right)~.
\]
After canceling $\Gamma$ and using the fact that $\cE_\alpha$ is decreasing, we obtain the same upper bound in both cases:
\[
\Gamma
\le
\frac{
b+\cE_\alpha^{-1}\!\left(b\Big/\Ex{s\sim\F_b}{s}\right)
}{
\Ex{s\sim\F_b}{s}
}
=
\frac{
b+\cE_\alpha^{-1}\!\left(\frac{\alpha b}{1+\alpha b}\right)
}{
b+\frac{1}{\alpha}
}~.
\]
Since $b\ge1$ was arbitrary, minimizing the right-hand side over $b$ proves that $\Gamma\le\overline{\Gamma}_\alpha$.

For MHR distributions, $\alpha=1$ and $\cE_1^{-1}(q)=-\ln q$. Consequently,
\[
\overline{\Gamma}_1
=
\min_{b\ge1}
\frac{b+\ln(1+1/b)}{b+1}
\approx 0.800654~ < 0.801.
\]

Finally, consider $\alpha=0$. In this case, $\F_b=\F^{+}_{1,b}$ has an infinite mean. If $\Gamma>0$, the concavity argument in \eqref{eq:concave-mean-lower-bound} would require $p(\F_b)=\infty$, contradicting that $p$ is a real-valued pricing policy. Hence, no strictly positive approximation ratio can be achieved over the family of regular distributions.

\end{proof}

\section{Robust Pricing with a Monotone Statistic}
\label{sec:application}

We now apply \Cref{thm:monotone} to robust pricing when the seller observes a distributional statistic. For monotone, homogeneous statistics, we establish the equivalence between statistic-based robust pricing and pricing policies that depend on that statistic, and scaling reduces the two-parameter boundary family to a one-parameter family. This reduction enable optimal guarantees for $L^\eta$-norm and superquantile information.

Consider a robust pricing setting where the seller observes a monotone statistic $\Psi(\F)=\theta$ of the value distribution. 
That is, the seller does not know the distribution $\F$, but only that it lies in the ambiguity set
\begin{align*}
    \cF (\theta)= \left\{\F \in \cF \mid \Psi(\F) = \theta \right\}.
\end{align*}
The quantile setting studied by \cite{ms/AllouahBB23} is special because the observed information is a local point constraint on the demand curve: fixing a quantile value \(w\) at tail probability \(q\) fixes the point \((w,q)\). Their extremal construction can therefore be built to pass through the observed point and the oracle-optimal point, so the constructed extremal distribution remains feasible for the original single-point ambiguity set. In contrast, a general statistic constraint \(\Psi(\F)=\theta\), such as a mean, moment, or superquantile constraint, imposes a global restriction on the distribution and need not fix any common point on the demand curve. An extremal distribution constructed by preserving selected demand-curve points will generally change the value of \(\Psi\), and hence may not remain in \(\cF(\theta)\). Monotonicity of \(\Psi\) under first-order stochastic dominance only determines the direction in which \(\Psi\) changes; it does not preserve the equality constraint \(\Psi(\F)=\theta\). Thus, the key feasibility-preservation step underlying the single-point reduction of \cite{ms/AllouahBB23} does not extend directly to general statistic-defined ambiguity sets.
Our approach instead uses homogeneity to rewrite the statistic-based problem in terms of the global price functional \(p(\F)=\omega\Psi(\F)\), applies \Cref{thm:monotone} to this moving price, and then uses scaling to return to a normalized statistic level set.

The seller's objective is to design a pricing policy that maximizes the worst-case approximation ratio over this ambiguity set. We further assume the statistic $\Psi$ is homogeneous.
\begin{definition}
    A functional/statistic $\Psi:\Delta(\R_+)\to\R$ is homogeneous if, for any distributions $\F_1$ and $\F_2$ satisfying $\F_1(\beta v)=\F_2(v)$ for all $v$, we have $\Psi(\F_1) = \beta \Psi(\F_2)$.
\end{definition}
Each of our motivating statistics---mean, $L^\eta$-norms, superquantiles, quantiles---indeed satisfies homogeneity. 
In this section, we solve the robust pricing problem under homogeneous statistics and compute the optimal approximation ratio. Specifically, the seller chooses a statistic-based price map \(\rho:\R\to\R\), and\ we are solving the following problem:
\begin{align}
    \max_{\rho:\R \to \R}\, \min_{\theta} \, \min_{\F\in \cF(\theta)} \frac{\Rev(\rho(\theta),\F)}{\OPT(\F)}~.
    \label{eq: policy scale invariant}
\end{align}
For instance, when $\Psi(\F) = \|\F \|_1$, problem \eqref{eq: policy scale invariant} becomes the robust pricing problem with mean information. We show that this robust pricing problem with mean information is equivalent to finding the optimal mean pricing policy in problem \eqref{eq: linear}, which is formalized in the following theorem.
\begin{theorem}
For any homogeneous statistic $\Psi$, the robust pricing problem under $\Psi$ information in \eqref{eq: policy scale invariant} is equivalent to the $\Psi$-pricing problem \eqref{eq: linear}:
\[
\max_{\rho:\R \to \R}\, \min_{\theta} \, \min_{\F\in \cF(\theta)} \frac{\Rev(\rho(\theta),\F)}{\OPT(\F)} = \max_{\omega} \min_{\F \in \cF} \frac{\Rev(\omega\cdot \Psi(\F), \F)}{\opt(\F)}~.
\]
Moreover, the optimal approximation ratio is achieved by a linear pricing policy $\rho^*(\theta) = \omega^* \cdot \theta$ where
$\omega^*$ is an optimizer of the right-hand side.
\label{thm:scale}
\end{theorem}

\Cref{thm:scale} establishes the equivalence between a robust pricing problem under statistic-based information and the corresponding optimization over statistic-based pricing policies. In addition, it demonstrates that when the statistic $\Psi$ is homogeneous, it is without loss of generality to restrict attention to homogeneous (linear in $\Psi$) pricing policies in the robust pricing problem. The key observation is that the optimal revenue function $\opt:\Delta(\R_+)\to\R$ is itself homogeneous in the valuation distribution. Consequently, the robust pricing problem reduces to optimizing the discount factor $\omega$ multiplying the statistic $\Psi(\F)$. Based on this simplification of the seller's problem, we further reduce the nature's decisions to one-parameter distribution families in \eqref{eq: scale ab} and \eqref{eq: scale psi}.

\begin{theorem}
For any monotone and homogeneous statistic $\Psi$, the robust pricing problem \eqref{eq: policy scale invariant} and the $\Psi$-pricing problem \eqref{eq: linear} are equivalent to the following two problems in \eqref{eq: scale ab} and \eqref{eq: scale psi}.
    \begin{align}
        \max_{\omega} \min_{\F \in \cF} \frac{\Rev(\omega\cdot \Psi(\F), \F)}{\opt(\F)} = \max_{\omega} \min_{\F \in\scFone \cup \lcFone} \frac{\Rev(\omega\cdot \Psi(\F), \F)}{\opt(\F)} 
        \label{eq: scale ab}
    \end{align}
        where $\scFone$, $\lcFone$ are the set of boundary distributions of $\alpha$-regular distributions with $a = 1$ or $b=1$, respectively:
\[
\scFone := \left\{ \F^{-}_{\lambda,1} \mid \lambda \ge 0, \, \alpha \lambda \le 1 \right\} \quad \text{and} \quad 
\lcFone := \left\{\F^{+}_{\lambda, 1} \mid  \lambda\ge 1 \right\}~.
\]
 At the same time,    
    \begin{align}
        \max_{\omega} \min_{\F \in \cF} \frac{\Rev(\omega\cdot \Psi(\F), \F)}{\opt(\F)}
        = \max_{\omega}\min_{\F\in \scF\of{1} \cup \lcF\of{1}} \frac{\Rev(\omega\cdot \Psi(\F),\F)}{\OPT(\F)}~
        \label{eq: scale psi}
    \end{align}
where $\scF\of{1} = \left\{\F \in \scF \mid \Psi(\F) =1 \right\}, \lcF\of{1} = \left\{\F \in \lcF \mid \Psi(\F) = 1  \right\}$.
\label{thm: scale-free}
\end{theorem}

\Cref{thm: scale-free} shows that the worst-case value in the robust-pricing problem can be evaluated over a one-parameter family, substantially improving tractability. The theorem provides two equivalent scaling formulations with different computational advantages.
For different forms of statistical information, either \eqref{eq: scale ab} or \eqref{eq: scale psi} can be applied to derive the optimal pricing policy and the corresponding approximation ratio. In the subsections that follow, we focus on MHR ($\alpha=1$) distributions. We characterize the optimal pricing policy and its approximation ratio under different representative information, including $L^\eta$-norm and superquantile information. We present proofs based on either \eqref{eq: scale ab} or \eqref{eq: scale psi}, depending on which approach yields greater clarity or notational simplicity. Each approach offers distinct analytical advantages, and we deliberately use both to illustrate alternative proof techniques.

\subsection{\texorpdfstring{$L^\eta$}{L-eta}-Norm Information}
 
In this subsection, we study the robust pricing problem in which the seller observes the $L^\eta$-norm information of the valuation distribution, $\|\F\|_\eta = \E[v^\eta]^{1/\eta}=\theta$, for some $\eta \ge 1$. 
The seller seeks a posted price that maximizes the worst-case approximation ratio over all distributions consistent with this information. In \Cref{prop: norm}, we provide a reformulation of the corresponding robust pricing problem with the help of \Cref{thm: scale-free}. Since $p(\F) = \E[v^\eta]^{1/\eta}$ can be implemented by the sample-contingent transfer rule\ $h(s,\F) = 1/\eta \cdot  \left((\eta-1)\|\F\|_{\eta}  + s^\eta / \|\F\|_{\eta}^{\eta-1} \right)$, \Cref{prop: norm} characterizes both the optimal robust price and its approximation ratio under $L^\eta$-norm information, as well as the performance guarantee of the sample-contingent mechanism\ defined above. 

We remark that when $\eta < 1$, the quantity $\E[v^\eta]^{1/\eta}$ is known as the quasinorm of $\F$. Although it is monotone, the quasinorm is a strictly convex---rather than concave---functional. Thus, it cannot be implemented by a sample-contingent mechanism. Nevertheless, the quasinorm remains meaningful distributional information in robust pricing, and \Cref{prop: norm} can still characterize the optimal approximation and price under a robust pricing scheme given quasinorm information with $\eta < 1$.
    \begin{proposition}
    When the seller knows the $L^\eta$-norm of the distribution, i.e., $\Psi(\F)= \|\F\|_\eta$ in Problem \eqref{eq: policy scale invariant}, the optimal discount factor $\omega$ and the corresponding approximation ratio $\Gamma$ can be solved by the following program.
    \begin{subequations} 
        \begin{align}
        \setcounter{equation}{-1}
            \max_ {\omega,\Gamma}: \quad & \Gamma \label{eq:norm-obj}\\
        \mbox{subject to} : \quad & \Gamma \cdot e^{-\lambda} \le \omega \cdot \|\sF[1]\|_\eta \cdot e^{-\lambda \omega \|\sF[1]\|_\eta} & \forall \lambda\in [0,1]\label{eq:norm-a}\\
      &   \Gamma \cdot 1 \le \omega \cdot \|\lF[1]\|_\eta \cdot \min \left\{ e^{-\lambda(\omega \|\lF[1]\|_\eta-1)}, 1\right\}
      & \forall \lambda\in [1,\infty) \label{eq:norm-b}
        \end{align}
        \end{subequations}
Here, $ \|\F^{-}_{\lambda,1}\|_\eta = (\lambda^{-\eta}\cdot \int_0^{\lambda} t^{\eta}e^{-t}dt + e^{-\lambda})^{1/\eta}$ 
and 
$\|\F^{+}_{\lambda,1}\|_\eta = (e^{\lambda }\cdot  \lambda^{-\eta} \cdot \int_{\lambda}^\infty t^{\eta}e^{-t}dt)^{1/\eta}.$
\label{prop: norm}
 \end{proposition}

\paragraph{Mean Information.} The $L^\eta$-norm information framework subsumes mean information as a special case when $\eta = 1$. In this case, we simplify the formulas of $\|\sF[1]\|_\eta$ and $\|\lF[1]\|_\eta$:
\begin{align*}
    \|\F^{-}_{\lambda,1}\|_1  = \frac{1}{\lambda} \cdot \int_0^{\lambda} te^{-t}dt + e^{-\lambda} = \frac{1-e^{-\lambda}}{\lambda}; \quad \|\lF[1]\|_1  = e^{\lambda} \cdot \frac{1}{\lambda} \int_\lambda^\infty t e^{-t} dt = 1 + \frac{1}{\lambda}.
\end{align*}
Consequently, the optimization problem \eqref{eq:norm-obj} simplifies to the following form.
        \begin{align*}
            \max_ {\omega,\Gamma}: \quad & \Gamma \\
        \mbox{subject to}:\quad
        & \Gamma \le \omega \cdot \frac{1-e^{-\lambda}}{\lambda} \cdot e^{-\omega (1-e^{-\lambda}) + \lambda} & \forall \lambda\in [0,1] \\
      & \Gamma \le \omega\cdot \of{1+\frac{1}{\lambda}}\cdot \min\offf{ e^{\lambda-\omega (\lambda+1)}, 1} & \forall \lambda\in [1,\infty)
        \end{align*}
Numerically solving the problem, we characterize the optimal price and its approximation ratio when the seller knows the mean of the distribution.
 \begin{proposition}
        When the seller knows $\|\F\|_1 = \theta$, then the optimal price is $\sim 0.823 \cdot \theta$, which achieves an approximation ratio of 0.771. 
        \label{prop:mean}
    \end{proposition}

Beyond the mean, by numerically evaluating the optimal solution to problem \eqref{eq:norm-obj}, we obtain the optimal prices and the corresponding approximation ratios for different $\eta$. In \Cref{fig:moment_mhr}, we plot the optimal discount factors and the corresponding approximation ratios as a function of $\eta$.
We observe that $\eta \approx 1.37$ achieves the largest approximation ratio among all $L^\eta$-norm pricing policies.
\begin{proposition}
    When the seller knows $\|\F\|_{1.37}=\theta$, then the optimal price is $\sim 0.805 \cdot \theta$, which achieves an approximation ratio of 0.774. 
\end{proposition}

\begin{figure}[t]
\centering
\captionsetup{justification=centering}
\caption{Performance Ratio and Optimal Discount with $L^\eta$-norm and Superquantile Information}
\label{fig:cvar-mhr}
\begin{subfigure}[t]{0.45\textwidth}
\includegraphics[width=\textwidth,keepaspectratio]{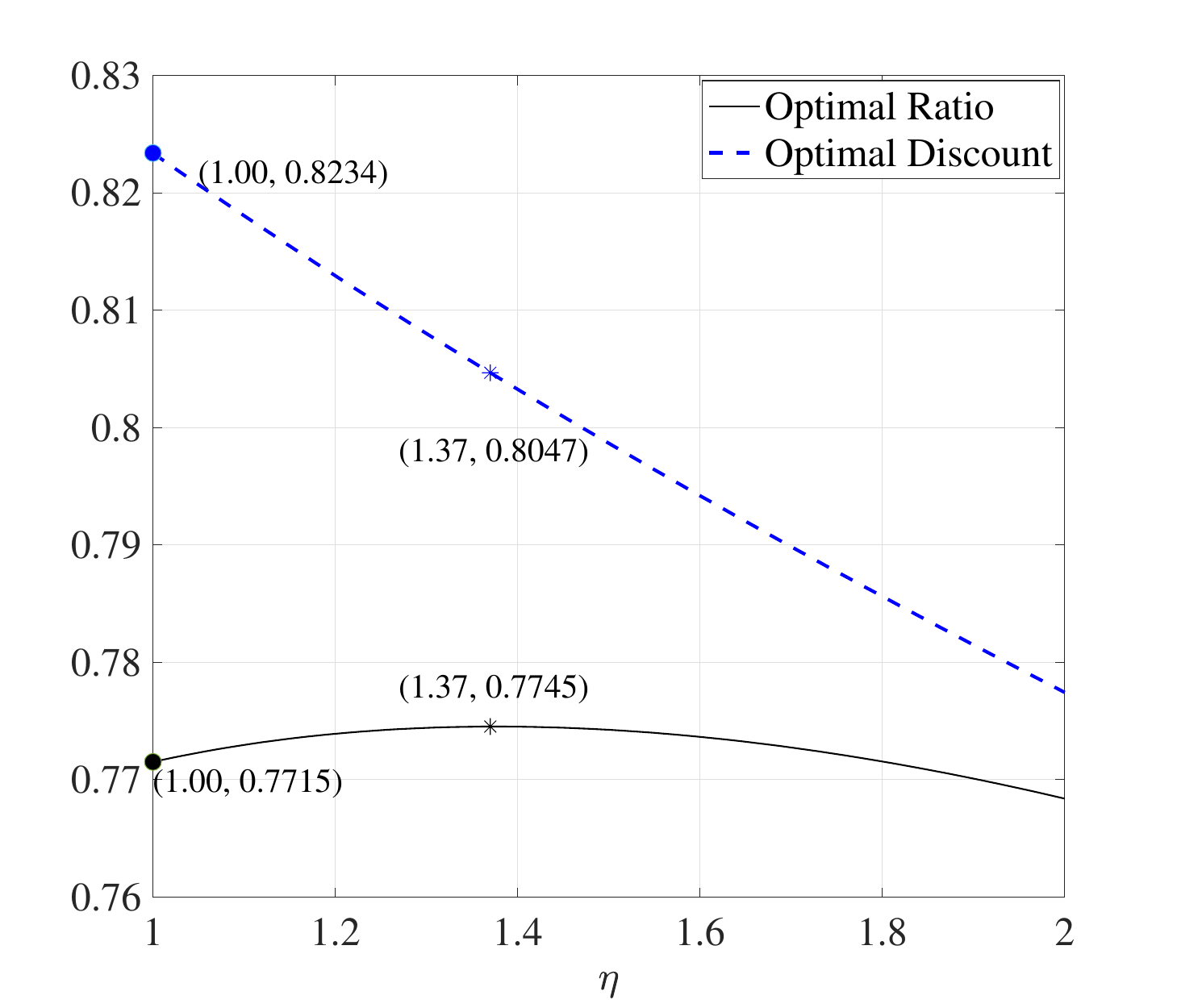}
\caption{$p(\F) = \omega \cdot \|\F\|_\eta = \omega \cdot \E[v^\eta]^{1/\eta}$}
\label{fig:moment_mhr}
\end{subfigure}
\begin{subfigure}[t]{0.45\textwidth}
\includegraphics[width=\textwidth,keepaspectratio]{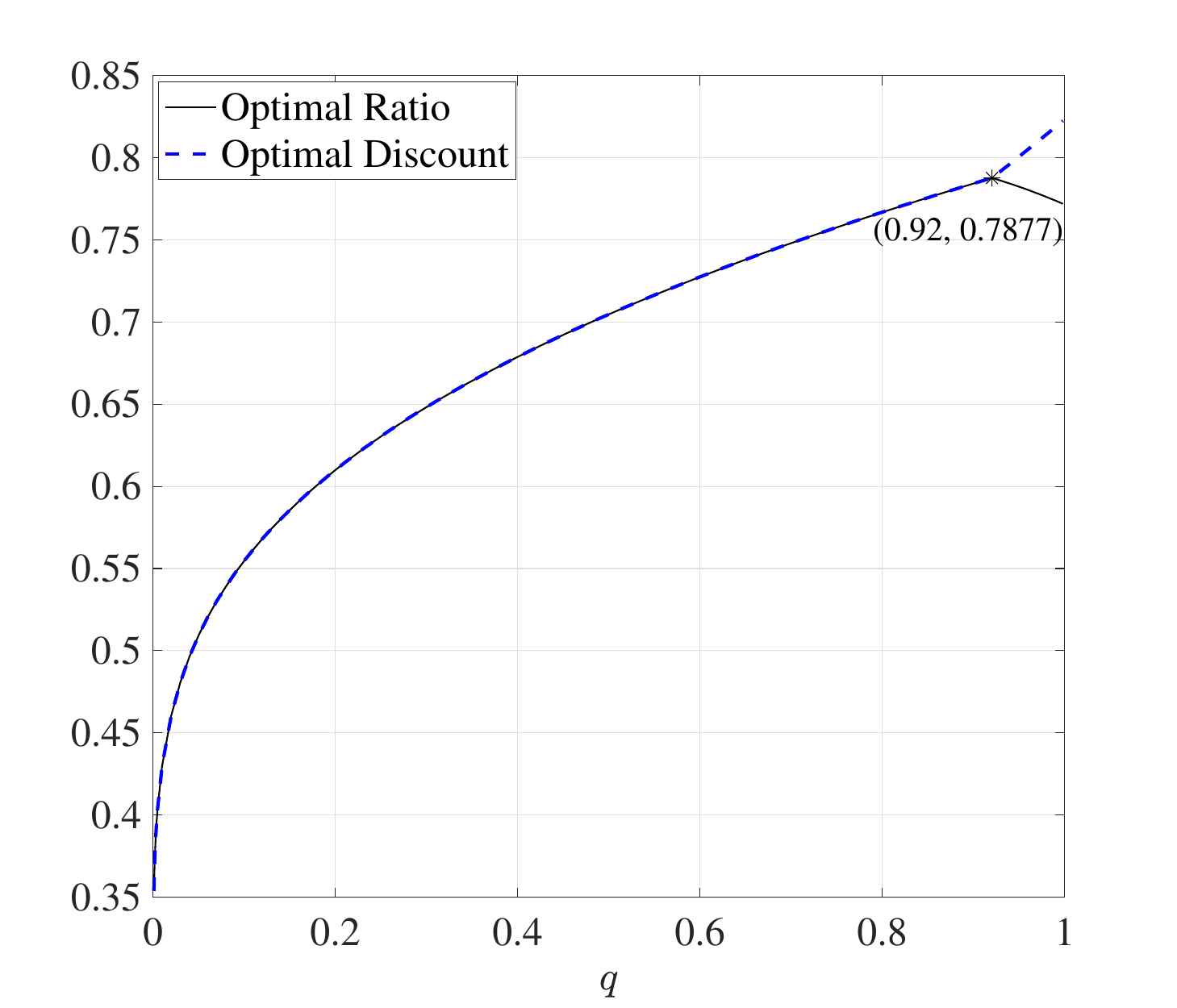}
\caption{$p(\F) = \omega \cdot \mathrm{CVaR}_{q}(\F)$}
\label{fig:fig_cvar_lowerq}
\end{subfigure}
\end{figure}

\subsection{Superquantile Information}
Suppose the seller observes aggregate information about only part of the market. For example, privacy-preserving dashboards or A/B tests may report the average value within a selected consumer segment. We use the atom-safe definitions of \(\mathrm{VaR}_{u}\) and \(\mathrm{CVaR}_{q}\) introduced in the Introduction.
A threshold \(x\) minimizes the variational expression whenever \(\Pr[s\sim\F]{s>x}\le q\le\Pr[s\sim\F]{s\ge x}\); for atomless distributions this reduces to the familiar mean of the top-\(q\) tail.\ Next, we investigate how the value of learning superquantile information\ changes with $q$. This provides a clean way to quantify the marginal benefit of collecting comprehensive aggregate market information versus more targeted market segment information. We remark that the superquantile information also subsumes mean information as a special case when $q=1$.

\begin{proposition}
\label{prop:superquantile}
 When the seller knows $\mathrm{CVaR}_{q}$, i.e., $\Psi(\F)= \mathrm{CVaR}_{q}(\F)$ in Problem \eqref{eq: policy scale invariant}, the optimal discount factor $\omega$ and the corresponding approximation ratio $\Gamma$ can be solved by:         \begin{subequations} 
        \label{eq:cvar}
        \begin{align}
        \setcounter{equation}{-1}
    \max_ {\omega,\Gamma}: \quad &  \Gamma \\
        \mbox{subject to}: \quad  &    \Gamma \le \omega& \label{eqn:cvar_1} \\
         &    \Gamma \le \frac{\omega\cdot e^{-\lambda \omega}}{q\cdot\of{1-\lambda-\ln q}\cdot -\frac{\ln\of{-q\cdot\of{\lambda+\ln q-1 }}}{\lambda} } & \forall \lambda\in (-\ln q, 1-\ln q -\frac{1}{eq}] \label{eqn:cvar_2}\\
      &   \Gamma \le \frac{\omega\cdot \min\offf{1, q e^{-\lambda \of {\omega-1+1/\lambda } }}}{1-\frac{1}{\lambda}\of{1-\ln q}} & \forall \lambda\in [2-\ln q,\infty ). \label{eqn:cvar_3}
        \end{align}
        \end{subequations}
\end{proposition}
Numerically solving the optimization problem above
computes the optimal approximation ratio and the corresponding optimal discount factor $\omega$ as a function of $\mathrm{CVaR}_q$. The results for different values of $q$ are depicted in \Cref{fig:fig_cvar_lowerq}, which indicates that the approximation ratio is maximized at approximately $q \approx 0.92$.

   \begin{proposition}
        When the seller knows $\mathrm{CVaR}_{0.92}=\theta$, then the optimal price is $\sim 0.787 \cdot \theta$, which achieves an approximation ratio of $0.787$. 
        \label{prop:cvar}
    \end{proposition}
\Cref{prop:cvar} suggests that selectively excluding low-value segments can yield stronger revenue guarantees than relying on aggregate means, which offers a practical implication for market research and data collection. Rather than estimating the mean over the entire population, it may be more effective to focus on estimating tail-based statistics that emphasize medium and high valuations.

\section{Conclusion}
We study robust pricing when the seller has only one sample from the valuation distribution but the buyer knows the distribution. Sample-contingent mechanisms exploit this informational asymmetry by requiring the buyer to participate and report before observing the sample. They allow the seller to implement pricing policies based on statistics such as the mean, $L^\eta$-norms, and superquantiles without knowing those statistics directly. For monotone policies over $\alpha$-regular distributions, we reduce the worst-case analysis to a two-parameter family of boundary distributions. This reduction yields the optimal monotone sample-contingent mechanism and sharp guarantees for several familiar pricing rules. Our impossibility results show that the monotonicity restriction sacrifices little performance for MHR distributions. Finally, the equivalence with statistic-based robust pricing places the sample-based and statistic-based approaches in a common framework and illustrates how the buyer's informational advantage can compensate for the seller's limited data.


\clearpage
\appendix

\section{\texorpdfstring{Missing Proofs from Section~\ref{sec:monotone}}{Missing Proofs from the Worst-Case Analysis}}
\label{app:proofs}

\begin{proof}[Proof of Lemma~\ref{lem:opt_price}]
We first study $\sF$ and calculate the derivative of the revenue function $\Rev(v,\sF)=v\,\barsF(v)=v\cE(\lambda v)$ on its continuous part $v\in[0,a)$:
\begin{align*}
\alpha \in [0,1): \quad \Rev'(v) & = \of{1+(1-\alpha)\lambda v}^{-\frac{1}{1-\alpha}} - \frac{1}{1-\alpha} \cdot (1-\alpha)\lambda v \cdot \of{1+(1-\alpha)\lambda v}^{-\frac{1}{1-\alpha}-1} \\
& = \of{1+(1-\alpha)\lambda v}^{\frac{\alpha-2}{1-\alpha}}\cdot\of{1 - \alpha \lambda v}~; \\
\alpha = 1: \quad \Rev'(v) & = e^{-\lambda v} - v\lambda e^{-\lambda v} = e^{-\lambda v}\cdot \of{1-\lambda v}~.
\end{align*}
In both cases, the revenue function is increasing in $v$ if $\alpha \lambda v\le 1$ and is decreasing in $v$ otherwise. 
Recall that the distribution $\sF$ is truncated at $a$. Hence, if \(\alpha=0\), then \(p^*=a\); if \(\alpha>0\), then \(p^*\) is either \(a\) or \(1/(\alpha\lambda)\), as stated.

Next, we study $\lF$ and calculate the derivative of the revenue function $\Rev(v,\lF)=v\,\barlF(v)=v\cE(\lambda (v-b))$ on its continuous part $v\in(b,\infty)$:
\begin{align*}
\alpha \in [0,1): \quad \Rev'(v)
&= \of{1+(1-\alpha)\lambda (v-b)}^{-\frac{1}{1-\alpha}} \\
&\quad -\frac{1}{1-\alpha} \cdot (1-\alpha) \lambda v
\cdot \of{1+(1-\alpha)\lambda (v-b)}^{-\frac{1}{1-\alpha}-1}\\
& = \of{1+(1-\alpha)\lambda(v-b)}^{\frac{\alpha-2}{1-\alpha}}\cdot \of{1-\alpha \lambda v - (1-\alpha) \lambda b}~; \\
\alpha = 1: \quad \Rev'(v) & = e^{-\lambda (v-b)} - \lambda v e^{-\lambda (v-b)} = e^{-\lambda (v-b)}\cdot \of{1-\lambda v}~.
\end{align*}
In both cases, the revenue function is increasing in $v$ if $1-\alpha \lambda v - (1-\alpha) \lambda b \ge 0$ and is decreasing in $v$ otherwise. For $\alpha>0$, this yields either $p^*=b$ or $p^*=\frac{1-(1-\alpha)\lambda b}{\alpha \lambda}$, as stated. For $\alpha=0$, the derivative has the constant sign of $1-\lambda b$: if $\lambda b\ge1$, then $p^*=b$, whereas if $0<\lambda b<1$, revenue increases to \(1/\lambda\) as \(v\to+\infty\), so \(p^*=+\infty\) and \(\Rev(p^*,\lF)=1/\lambda=\opt(\lF)\).

\end{proof}

\begin{proof}[Proof of Lemma~\ref{lem:lb-ub}]
At finite prices, the weak survival function is left-continuous and has only downward right jumps, so \(v\mapsto v\barf(v)\) is upper semicontinuous. At \(+\infty\), upper semicontinuity follows from the limit that defines \(\Rev(+\infty,\F)\). Hence extended revenue is upper semicontinuous on \(\overline{\R}_+\), and every extended revenue superlevel set is closed.
On the continuous interior of the support, the revenue function is $\Rev(v,\F)=v\,\barf(v)$. Its derivative equals
\[
\frac{\partial}{\partial v}\Rev(v,\F) = \barf(v)-v f(v) = -f(v) \cdot \varphi_{0,\F}(v)~.
\]
Hence, $\frac{\partial}{\partial v}\Rev(v,\F) \ge 0 \iff \varphi_{0,\F}(v) \le 0$.
Since $\varphi_{0,\F}(v) = \varphi_{\alpha,\F}(v) + \alpha v$ and $\varphi_{\alpha,\F}$ is increasing by $\alpha$-regularity, it follows that $\varphi_{0,\F}$ is also increasing.
Thus, on the continuous interior, \(\Rev(v,\F)\) is non-decreasing until \(\varphi_{0,\F}\) crosses zero and non-increasing thereafter. If no crossing occurs before a finite upper endpoint, the weak-survival convention includes the endpoint atom and the revenue maximum occurs there. If revenue instead remains non-decreasing on all finite prices, its value at \(+\infty\) is the limiting value defined above. Consequently, every extended revenue superlevel set is a closed interval in \(\overline{\R}_+\); it is nonempty because it contains \(p^*(\F)\).
The statement thus holds by the definition of $\mathrm{LB}$ and $\mathrm{UB}$.

\end{proof}

\begin{proof}[Proof of Lemma~\ref{lem:lb-lipschitz}]
We first verify that $v \to v \cE(\lambda v)$ is log-concave on the region when $\alpha \lambda v \le 1$:
\begin{align*}
\alpha \in [0,1): \quad \left( \ln \left(v \cE(\lambda v) \right) \right)' & = \left( \ln v - \frac{1}{1-\alpha} \ln(1+(1-\alpha)\lambda v) \right)' = \frac{1}{v} - \frac{\lambda}{1+(1-\alpha)\lambda v}~; \\
\alpha = 1: \quad (\ln \left(v \cE(\lambda v) \right))' & = (\ln v - \lambda v)' = \frac{1}{v} - \lambda~.
\end{align*}
We further calculate the second derivative:
\begin{align*}
   \left(\frac{1}{v} - \frac{\lambda}{1+(1-\alpha)\lambda v}\right)' & = - \frac{1}{v^2} + \frac{(1-\alpha) \lambda^2}{(1+(1-\alpha)\lambda v)^2} \\
   & = -\frac{1+2(1-\alpha)\cdot \lambda v - \alpha (1-\alpha) \lambda^2v^2}{v^2(1+(1-\alpha)\lambda v)^2} \le 0~, \qquad \forall \alpha \lambda v \le 1
\end{align*}
Next, suppose $a \le a'$, and let $l = \lb$. Define $l' = l + a' - a$. It suffices to prove that $\Rev(l', \sF[a']) \ge \Gamma \cdot \opt(\sF[a'])$:
\begin{align*}
\Rev(l', \sF[a']) = l' \cE(\lambda l') & \ge \left(l \cE(\lambda l)\right) \cdot \left(a' \cE(\lambda a')\right) / \left(a \cE(\lambda a)\right) \\
& = \Gamma \cdot a' \cE(\lambda a') = \Gamma \cdot \opt(\sF[a'])~,
\end{align*}
where the inequality holds by the log-concavity of $v \to v\cE(\lambda v)$ and the fact that $l \le a,l' \le a'$ and $l+a' = l'+a$; the last equality holds by Lemma~\ref{lem:opt_price}.

\end{proof}

\subsection{\texorpdfstring{Proofs for the Worst-Case Reduction}{Proofs for the Worst-Case Reduction}}
\label{app:reduction-proofs}

\begin{proof}[Proof of \Cref{lem:small}]
First suppose that \(p^*(\F)<+\infty\), and set \(v_0=p^*(\F)\) and \(q(v):=\barf(v)\). Consider the distribution \(\sF\) with parameters
\[
\lambda=\cE^{-1}(q(v_0))/v_0
\qquad\text{and}\qquad
a=v_0.
\]
This is the lower \(\alpha\)-regular boundary distribution consistent with the point \(\bigl(v_0,\barf(v_0)\bigr)\).

We first verify that \(\sF\in\scF\), or equivalently \(\alpha\lambda a\le1\). This is automatic when \(\alpha=0\). When \(\alpha>0\), \Cref{lem:selling_prob} gives \(q(v_0)\ge\alpha^{1/(1-\alpha)}\), where the right-hand side is interpreted as \(1/e\) for \(\alpha=1\). Hence
\[
\cE^{-1}(q(v_0))\le\frac1\alpha
\qquad\text{and thus}\qquad
\alpha\lambda a=\alpha\cE^{-1}(q(v_0))\le1.
\]

Next, \(\sF\distle\F\). For \(v>v_0\), \(\barsF(v)=0\le\barf(v)\). For \(0\le v\le v_0\), convexity of \(\cE^{-1}(q(v))\), together with \(q(0)=1\) and \(\cE^{-1}(1)=0\), gives
\begin{align*}
\barsF(v)
&=\cE\!\left(\cE^{-1}(q(v_0))\frac{v}{v_0}
+\cE^{-1}(q(0))\frac{v_0-v}{v_0}\right)\\
&\le \cE\!\left(\cE^{-1}(q(v))\right)=q(v)=\barf(v),
\end{align*}
because \(\cE\) is non-increasing. The same calculation covers a finite upper-endpoint atom because \(q(v_0)=\barf(v_0)\) uses weak survival. Therefore \(\Rev(v,\sF)\le\Rev(v,\F)\) for every \(v\), while
\[
\opt(\sF)=\Rev(v_0,\sF)=\Rev(v_0,\F)=\opt(\F).
\]
Monotonicity gives \(p(\sF)\le p(\F)\le v_0\). Because the revenue of \(\F\) is non-decreasing up to its leftmost monopoly price \(v_0\),
\[
\frac{\Rev(p(\sF),\sF)}{\opt(\sF)}
\le
\frac{\Rev(p(\sF),\F)}{\opt(\F)}
\le
\frac{\Rev(p(\F),\F)}{\opt(\F)}.
\]
Thus the conclusion holds with \(\varepsilon=0\) whenever \(p^*(\F)<+\infty\).

Now suppose that \(p^*(\F)=+\infty\). Then \(\alpha=0\), the revenue function is non-decreasing on finite prices, and the extended-price convention gives
\[
\Rev(+\infty,\F)=\lim_{a\to+\infty}\Rev(a,\F)=\opt(\F).
\]
Although \(+\infty\) is an optimizer in the extended price space, the family \(\scF\) has finite truncation points. We therefore use finite anchors approaching this optimizer. Given \(\varepsilon>0\), choose a finite \(a_\varepsilon\ge p(\F)\) sufficiently large that
\[
\frac{\Rev(p(\F),\F)}{\Rev(a_\varepsilon,\F)}
\le
\frac{\Rev(p(\F),\F)}{\opt(\F)}+\varepsilon.
\]
Let \(\F_\varepsilon^-\) have parameters
\[
\lambda_\varepsilon=\frac{\cE^{-1}(\barf(a_\varepsilon))}{a_\varepsilon}
\qquad\text{and}\qquad a=a_\varepsilon.
\]
The preceding chord argument gives \(\F_\varepsilon^-\distle\F\). Moreover, \(p^*(\F_\varepsilon^-)=a_\varepsilon\) because \(\alpha=0\), while monotonicity gives \(p(\F_\varepsilon^-)\le p(\F)\le a_\varepsilon\). Repeating the preceding ratio comparison with \(\Rev(a_\varepsilon,\F)=\opt(\F_\varepsilon^-)\) proves the claimed \(\varepsilon\)-bound.

\end{proof}

\begin{proof}[Proof of \Cref{lem:large}]
Because \(p(\F)\) is finite, the hypothesis \(p(\F)>p^*(\F)\) implies \(p^*(\F)<+\infty\). Write \(q(v):=\barf(v)\), \(g(v):=\cE^{-1}(q(v))\), and \(v_0:=p(\F)\). We first give the construction at any point \(x\) in the continuous part of the support satisfying \(p^*(\F)\le x\le v_0\) and \(q(x)>0\). Since the revenue of \(\F\) is non-increasing to the right of \(x\), right difference quotients of \(v\cE(g(v))\) at \(x\) imply
\[
g'_+(x)x\ge1+(1-\alpha)g(x).
\]
For \(\alpha=1\), the right-hand side is interpreted as \(1\). Choose a supporting slope \(\lambda\in\partial g(x)\) satisfying this inequality---for example, \(\lambda=g'_+(x)\)---and set
\[
b:=x-\frac{g(x)}{\lambda}.
\]
This construction uses no value of the density at \(x\). On the continuous interior, whenever the ordinary derivative exists, it reduces to \(\lambda=g'(x)=f(x)q(x)^{\alpha-2}\).

Let \(\lF\) be the boundary distribution with parameters \((\lambda,b)\). The affine function \(\ell(v):=g(x)+\lambda(v-x)=\lambda(v-b)\) supports the convex function \(g\), so \(\ell(v)\le g(v)\). Hence
\[
\barlF(v)=
\begin{cases}
1, & v\le b,\\
\cE(\ell(v)), & v>b
\end{cases}
\ge\barf(v),
\qquad
\barlF(x)=\barf(x).
\]
Thus \(\lF\distge\F\), \(\opt(\lF)\ge\opt(\F)\), and monotonicity gives \(p(\lF)\ge p(\F)=v_0\). Moreover,
\[
\lambda x-(1-\alpha)g(x)\ge1,
\]
so \Cref{lem:opt_price} gives \(p^*(\lF)\le x\). Consequently,
\[
\Rev(p(\lF),\lF)\le\Rev(v_0,\lF).
\]
If \(x=v_0\), the weak-survival equality at \(v_0\) yields
\[
\frac{\Rev(p(\lF),\lF)}{\opt(\lF)}
\le
\frac{\Rev(v_0,\F)}{\opt(\F)}.
\]
The argument also applies when \(v_0\) is a finite upper endpoint carrying an atom and \(g\) has a finite left supporting slope there: choose \(\lambda\) at least as large as that one-sided slope and large enough that \(\lambda b=\lambda v_0-g(v_0)\ge1\). No density is evaluated at the atom.

It remains to put the constructed distribution in \(\lcF\) when \(\lambda b<1\). This case is impossible for \(\alpha=0\), because then \(\lambda b=\lambda x-g(x)\ge1\). For \(\alpha>0\), let
\[
v_1:=p^*(\lF)=\frac{1-(1-\alpha)\lambda b}{\alpha\lambda},
\qquad
q_1:=\barlF(v_1)=\cE(\lambda(v_1-b)),
\]
and construct \(\llF\) with parameters
\[
\lambda^\dagger:=\lambda q_1^{1-\alpha},
\qquad
b^\dagger:=v_1.
\]
Then \(\lambda^\dagger b^\dagger=1\), so \(\llF\in\lcF\), and \(\llF\distge\lF\). The proportional-tail identity is as follows.
\begin{claim}
For every \(v\ge v_1\),
\[
\barlF(v)=q_1\barllF(v).
\]
\label{claim1}
\end{claim}
It follows that \(\opt(\lF)=q_1\opt(\llF)\). Since \(p(\llF)\ge p(\lF)\ge v_0\ge x\ge v_1\) and the revenue of \(\llF\) is non-increasing above \(v_1\),
\[
\frac{\Rev(p(\llF),\llF)}{\opt(\llF)}
\le
\frac{\Rev(p(\lF),\llF)}{\opt(\llF)}
=
\frac{\Rev(p(\lF),\lF)}{\opt(\lF)}.
\]
Thus normalization never increases the performance ratio. Write \(\F_x^\dagger\in\lcF\) for the raw construction when \(\lambda b\ge1\) and for its normalized version otherwise.

We finally cover the cases not handled by a finite supporting slope at \(v_0\). In all such cases the relevant upper endpoint \(U:=\overline v_{\F}\) is finite. If \(v_0=U\) carries an atom but \(g'_-(U)=+\infty\), choose continuous-support points \(x\) sufficiently close to \(U\). Because the supporting lines at \(x\) converge at \(U\) to the finite endpoint value \(g(U)=\cE^{-1}(\barf(U))\), the construction above gives \(\F_\varepsilon^+\in\lcF\) satisfying
\[
\frac{\Rev(p(\F_\varepsilon^+),\F_\varepsilon^+)}{\opt(\F_\varepsilon^+)}
\le
\frac{\Rev(v_0,\F)}{\opt(\F)}+\varepsilon.
\]
If \(q(v_0)=0\), then \(v_0\ge U\) and the ratio on the right is zero. When there is no endpoint atom, choose \(x<U\) sufficiently close to \(U\); the raw supporting-line construction then has arbitrarily small weak survival at \(v_0\), and normalization cannot increase its ratio. When \(U\) carries an atom and \(v_0>U\), use a supporting line at \(U\) with a sufficiently large slope if a finite endpoint slope exists; otherwise choose \(x<U\) sufficiently close to \(U\). In each case, the raw construction dominates \(\F\), so monotonicity places its policy price at least at \(v_0\), its revenue is non-increasing there, and its optimal revenue is at least \(\opt(\F)\); normalization can only lower the ratio. This proves the lemma.

\end{proof}

\begin{proof}[Proof of \Cref{claim1}]
    For $\alpha\in [0,1)$, 
\begin{align*}
q_1 \cdot \barllF(v) & = q_1 \cdot \cE(\lambda^\dagger(v-b^\dagger))= q_1 \cdot \left( 1+\of{1-\alpha}\lambda^\dagger (v-b^\dagger) \right) ^{-\frac{1}{1-\alpha}} \\
& = \left( q_1^{-(1-\alpha)}+(1-\alpha) \lambda  \left(v-v_1 \right) \right)^{-\frac{1}{1-\alpha}} \\
& = \left( \lambda v_1 + (1-\alpha) \lambda  \left( v- v_1 \right) \right)^{-\frac{1}{1-\alpha}} = \left( (1-\alpha) \lambda v + \alpha \lambda v_1\right)^{-\frac{1}{1-\alpha}}  \\
& = \left( 1 + (1-\alpha) \lambda (v-b) \right)^{-\frac{1}{1-\alpha}} = \barlF(v).
\end{align*}
    For $\alpha = 1$, 
     \begin{align*}
     \barllF(v) = & \cE(\lambda^\dagger(v-b^\dagger))= e^{-\lambda (v-\frac{1}{\lambda})} = e^{1-\lambda b} \cdot e^{-\lambda (v-b)} = \frac{1}{q_1}\of{ \barlF(v) } .
    \end{align*}
    Thus, for all $\alpha\in [0,1]$ and all $v\ge v_1$, we have $\barlF(v) = q_1 \cdot \barllF(v)~.$

\end{proof}

\begin{proof}[Proof of \Cref{cor:monotone}]
The necessity holds by Lemma~\ref{lem:lb-ub} and the fact that $\sF,\lF \in \cF$. 

Next, we prove the sufficiency of the two conditions by contradiction. Suppose there exists a monotone pricing policy $p$ satisfying the two conditions, but is not $\Gamma$-approximate. Then, there exists a distribution $\F \in \cF$ such that
\(
\Rev(p(\F), \F) < \Gamma \cdot \opt(\F)~.
\)
Choose
\[
0<\varepsilon<\Gamma-\frac{\Rev(p(\F),\F)}{\opt(\F)}.
\]
If \(p(\F)\le p^*(\F)\), \Cref{lem:small} supplies \(\F_\varepsilon^-\in\scF\) with performance ratio below \(\Gamma\) and \(p(\F_\varepsilon^-)\le p^*(\F_\varepsilon^-)\). By \Cref{lem:lb-ub}, \(p(\F_\varepsilon^-)<\mathsf{LB}(\F_\varepsilon^-,\Gamma)\), contradicting the first condition. If \(p(\F)>p^*(\F)\), \Cref{lem:large} supplies \(\F_\varepsilon^+\in\lcF\) with performance ratio below \(\Gamma\) and \(p(\F_\varepsilon^+)\ge p^*(\F_\varepsilon^+)\). The same lemma gives \(p(\F_\varepsilon^+)>\mathsf{UB}(\F_\varepsilon^+,\Gamma)\), contradicting the second condition.

\end{proof}

\section{Proof of Proposition~\ref{prop: optimal numerical}}
\label{sec:numerical}

The detailed certificate in this appendix concerns the MHR case, corresponding to $\alpha=1$ and $\mathcal E(x)=e^{-x}$, and establishes the lower bound $\Gamma=0.791$. We also perform analogous computations on an $\alpha$-grid covering the full interval $[0,1]$. Figure~\ref{fig:alpha-ratio} plots the independently certified lower bounds at the grid points, each valid uniformly over the continuous distributional parameter space, together with the concave-policy upper bound from Theorem~\ref{thm:negative-concave}. For $0<\alpha\leq0.25$, the lower bound is certified by an exact linear witness of the form $h(s)=cs$; for $\alpha>0.25$, it is certified using a piecewise witness. \footnote{The implementation is available at \url{https://github.com/tomtao26/Pricing-with-an-Informed-Buyer}.} The remainder of this appendix provides the complete certification argument for the MHR endpoint $\alpha=1$.

The MHR proof has two logically separate components.  We first derive a
finite collection of inequalities that is sufficient for all values of the
continuous parameters $(a,\lambda,b)$.  We then verify this finite
collection using explicit witnesses.

\paragraph{Reduction and scalar bounds.}
We begin by reducing the feasibility conditions to normalized exponential
families and expressing their right-hand sides through scalar benchmark
prices.  This reduction isolates the two continuous inequalities that the
certificate must control.
As established in Sections~\ref{sec:monotone} and~\ref{sec:optimal}, it
suffices to consider distributions of the form
$\F^{+}_{1,b}$ with $b\geq1$.\footnote{For a general
$\F^{+}_{\lambda,b}$, set $\widetilde b=\lambda b$ and
$\widetilde s=\lambda s$.  If $h^{\widetilde b}$ is the normalized
witness, use
$h^\Gamma(s,\F^{+}_{\lambda,b})
=\lambda^{-1}h^{\widetilde b}(\lambda s)$.  Both $\mathsf{LB}$ and
$\mathsf{UB}$ are homogeneous under this rescaling.}
For each such distribution, we seek a non-decreasing function
$h^\Gamma(\cdot,\F^{+}_{1,b})$ satisfying
\begin{align}
\E_{s\sim\F^{+}_{1,b}}
  \left[h^\Gamma(s,\F^{+}_{1,b})\right]
&\leq \mathsf{UB}(\F^{+}_{1,b},\Gamma),
\label{cond:mhr-upper}\\
\E_{s\sim\F^{-}_{\lambda,a}}
  \left[h^\Gamma(s,\F^{+}_{1,b})\right]
&\geq \mathsf{LB}(\F^{-}_{\lambda,a},\Gamma),
\quad \forall\,\lambda,a\geq0:\lambda a\leq1.
\label{cond:mhr-lower}
\end{align}

Fix $0<\Gamma\leq1$.  Let $\mathrm{W}_0$ and $\mathrm{W}_{-1}$ denote
the principal and lower real branches of the Lambert W function,
respectively.  Directly from the definitions of the two benchmark prices,
\begin{align}
\mathsf{LB}(\F^{-}_{\lambda,a},\Gamma)
&=
\begin{cases}
\displaystyle
-\frac{1}{\lambda}\,
\mathrm{W}_0\!\left(-\Gamma\lambda a e^{-\lambda a}\right),
&\lambda>0,\\[7pt]
\Gamma a,&\lambda=0,
\end{cases}
\qquad \lambda a\leq1,
\label{eq:mhr-lb-scalar}\\
\mathsf{UB}(\F^{+}_{1,b},\Gamma)
&=-\mathrm{W}_{-1}\!\left(-\Gamma b e^{-b}\right),
\qquad b\geq1.
\label{eq:mhr-ub-scalar}
\end{align}
For $\lambda>0$, the value in \eqref{eq:mhr-lb-scalar} is the smaller
solution $q\in[0,a]$ of $q e^{-\lambda q}=\Gamma a e^{-\lambda a}$;
the expression at $\lambda=0$ is its continuous limit.  The value in
\eqref{eq:mhr-ub-scalar} is the unique solution $p\geq b$ of $p e^{-(p-b)}=\Gamma b$.
Thus, the two Lambert W branches select the lower and upper prices directly
that attain a $\Gamma$ fraction of the corresponding optimal revenue.

\begin{figure}[htbp]
\centering
\captionsetup{justification=centering}
\caption{Certified lower bound and concave-policy upper bound over
$\alpha$-regular distributions.  Every marker on the solid lower curve is
independently verified over the parameter space,
and the solid line is a monotone visual interpolation through these markers.
The dashed curve is the upper bound $\overline{\Gamma}_{\alpha}$ from
Theorem~\ref{thm:negative-concave}.}
\label{fig:alpha-ratio}
\includegraphics[width=0.45\textwidth,keepaspectratio]
{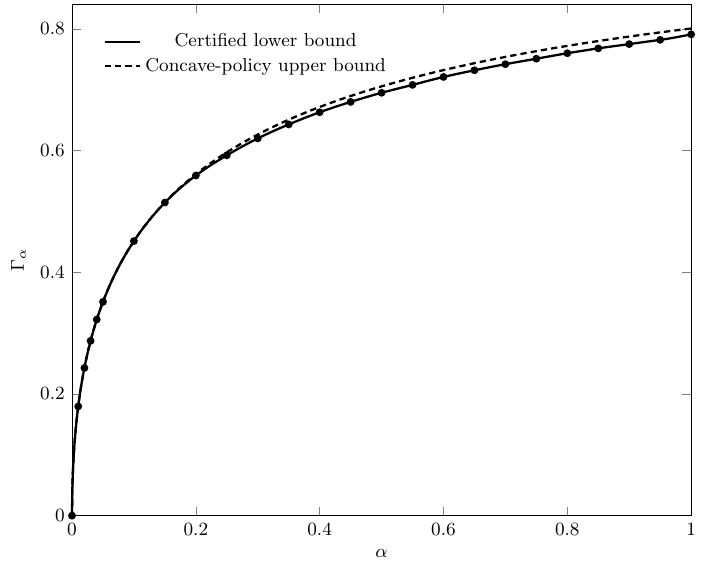}
\end{figure}

\paragraph{The transfer-schedule grid.}
We next restrict the transfer schedule to a finite step-function grid followed
by a linear tail.  The step values become the variables of the finite
certificate, while the tail will allow the lower constraints to be extended
beyond the last breakpoint.
Let
\[
1=A_0<A_1<\cdots<A_N=\bar a=20
\]
be the following grid.  We set $A_i=1+0.025i$ for $0\leq i\leq440$.
Starting from $A_{440}=12$, we recursively set
$A_{i+1}=\min\{1.05A_i,20\}$ until the endpoint $20$ is reached.  This gives
$N=451$ and hence $452$ breakpoints.

For a fixed $b$, write $h_i^b=h^b(A_i)$ and
$d_i^b=h_i^b-h_{i-1}^b$.  We consider functions of the form
\begin{equation}
h^b(s)=
\begin{cases}
h_0^b,&s\leq A_0,\\
h_i^b,&A_{i-1}<s\leq A_i,\quad 1\leq i\leq N,\\
h_N^b+e(s-\bar a),&s>\bar a.
\end{cases}
\label{eq:mhr-step-function}
\end{equation}
The conditions $h_0^b\geq0$ and $d_i^b\geq0$ ensure that $h^b$ is
nonnegative and non-decreasing.

For $1\leq i\leq N$ and
$\lambda\in[0,1/A_{i-1}]$, define
\begin{align}
L_i^b(\lambda)
&:=h_0^b+\sum_{j=1}^{i}d_j^b e^{-\lambda A_{j-1}},
\label{eq:mhr-cell-lhs}\\
R_i(\lambda)
&:=
\begin{cases}
\mathsf{LB}(\F^{-}_{\lambda,A_i},\Gamma),
   &0\leq\lambda\leq 1/A_i,\\
\mathsf{LB}(\F^{-}_{\lambda,1/\lambda},\Gamma),
   &1/A_i<\lambda\leq1/A_{i-1}.
\end{cases}
\label{eq:mhr-cell-rhs}
\end{align}

The next lemma removes the continuum of values of $a$ inside each grid cell.
It does so by moving the truncation point to the right endpoint of the cell
and encoding the largest possible lower benchmark in $R_i$.
\begin{lemma}
\label{lem:mhr-cell-reduction}
Suppose that $h^b$ has the form \eqref{eq:mhr-step-function}.  If
\[
L_i^b(\lambda)\geq R_i(\lambda),
\qquad
\lambda\in[0,1/A_{i-1}],\quad 1\leq i\leq N,
\]
and $h_0^b\geq\Gamma$, then \eqref{cond:mhr-lower} holds for every
$a\leq\bar a$.
\end{lemma}

\begin{proof}
Fix $a\in(A_{i-1},A_i]$ and $\lambda a\leq1$.  Because $h^b$ is
constant on $(A_{i-1},A_i]$, moving the atom of
$\F^{-}_{\lambda,a}$ from $a$ to $A_i$ and filling the intervening
interval with the same value of $h^b$ does not change the expectation.
Consequently,
\[
\E_{\F^{-}_{\lambda,a}}[h^b(s)]
=\E_{\F^{-}_{\lambda,A_i}}[h^b(s)]
=h_0^b+\sum_{j=1}^{i}d_j^b e^{-\lambda A_{j-1}}
=L_i^b(\lambda).
\]
For fixed $\lambda>0$, the map
$x\mapsto x e^{-\lambda x}$ is non-decreasing on
$[0,1/\lambda]$.  Thus, both the target
$\Gamma a e^{-\lambda a}$ and its smaller inverse image are
non-decreasing in $a$.  The characterization following
\eqref{eq:mhr-lb-scalar} therefore shows that
$a\mapsto\mathsf{LB}(\F^{-}_{\lambda,a},\Gamma)$ is
non-decreasing on this interval.  Hence, its largest value over
$(A_{i-1},A_i]\cap[0,1/\lambda]$ equals the first branch of
\eqref{eq:mhr-cell-rhs} when $\lambda A_i\leq1$, and the second branch
otherwise.  The case $\lambda=0$ follows by continuity.  Thus,
$R_i(\lambda)$ upper-bounds the required lower price throughout the cell.

It remains to consider $a\leq A_0=1$.  In this case,
$\E_{\F^{-}_{\lambda,a}}[h^b(s)]=h_0^b$.  Let
$q=\mathsf{LB}(\F^{-}_{\lambda,a},\Gamma)$.  For $\lambda>0$,
we have $q\leq a$ and
\[
q=\Gamma a e^{-\lambda(a-q)}\leq\Gamma a.
\]
The same inequality holds with equality when $\lambda=0$.  Consequently,
$\mathsf{LB}(\F^{-}_{\lambda,a},\Gamma)
\leq\Gamma a\leq\Gamma\leq h_0^b$.

\end{proof}

\paragraph{A finite sufficient condition for the continuous
$\lambda$-constraints.}
The finite program checks each lower constraint only on a grid of
$\lambda$-values.  We make this discretization rigorous by bounding the
curvature between adjacent grid points and reserving a corresponding
interpolation buffer.

The first step is a uniform second-derivative bound.  The following lemma
controls both the scalar lower benchmark and the difference between the two
sides of each cell constraint.
\begin{lemma}
\label{lem:mhr-interpolation}
Let $0<\Gamma\leq0.8$.  For every $z\in[0,1]$,
\[
\left|
\frac{\dd^2}{\dd z^2}
\left(
\frac{\mathsf{LB}(\F^{-}_{1,z},\Gamma)}{z}
\right)
\right|
\leq6,
\]
where the quotient at $z=0$ is defined by its continuous limit $\Gamma$.
Consequently, on each of the two smooth pieces in
\eqref{eq:mhr-cell-rhs},
\begin{equation}
\left|\frac{\dd^2}{\dd\lambda^2}
  \left(L_i^b(\lambda)-R_i(\lambda)\right)\right|
\leq
M_i^b
:=A_i^2(h_i^b-h_0^b)+6A_i^3.
\label{eq:mhr-second-derivative}
\end{equation}
\end{lemma}

\begin{proof}
To carry out the differentiation, write temporarily
\[
r(z):=\frac{\mathsf{LB}(\F^{-}_{1,z},\Gamma)}{z}
\quad\text{for }z>0,
\qquad r(0):=\Gamma.
\]
The defining identity
$z r(z)e^{-z r(z)}=\Gamma z e^{-z}$ gives, after canceling $z$ and
taking logarithms,
\[
\log r(z)-z r(z)=\log\Gamma-z.
\]
Implicit differentiation yields
\begin{equation}
r'(z)
=\frac{r(z)(r(z)-1)}
       {1-zr(z)}.
\label{eq:mhr-r-derivative}
\end{equation}
As shown in the proof of Lemma~\ref{lem:mhr-cell-reduction},
$z\mapsto\mathsf{LB}(\F^{-}_{1,z},\Gamma)=zr(z)$ is
non-decreasing on $[0,1]$.  Furthermore,
\[
\Gamma/e\leq0.8/e
<\tfrac12 e^{-1/2},
\]
so
$\mathsf{LB}(\F^{-}_{1,z},\Gamma)
\leq\mathsf{LB}(\F^{-}_{1,1},\Gamma)<1/2$.
It follows that $1-zr(z)\geq1/2$.  Moreover, the lower-root selection
gives $r(z)\leq1$, and hence
$r(z)=\Gamma e^{-z(1-r(z))}\leq\Gamma\leq0.8$.
Equation~\eqref{eq:mhr-r-derivative} therefore implies
$|r'(z)|\leq1/2$.
Differentiating once more gives
\[
r''(z)
=
\frac{
\left(2r-1-zr^2\right)r'
+r^2(r-1)}
{\left(1-zr\right)^2}.
\]
Using $r\leq0.8$, $|r'|\leq1/2$, and $1-zr\geq1/2$ shows that the
absolute value of the last display is at most
\[
\frac{1.64(1/2)+0.64}{(1/2)^2}=5.84<6.
\]
The Lambert W expression in \eqref{eq:mhr-lb-scalar} gives an analytic
extension of $r$ at $z=0$, so the same bound holds there.

Next,
\[
\left|(L_i^b)''(\lambda)\right|
=\sum_{j=1}^{i}d_j^b A_{j-1}^2 e^{-\lambda A_{j-1}}
\leq A_i^2(h_i^b-h_0^b).
\]
On $[0,1/A_i]$, we have
$R_i(\lambda)=A_i r(\lambda A_i)$ and hence
$|R_i''(\lambda)|\leq6A_i^3$.
On $[1/A_i,1/A_{i-1}]$, we have
$R_i(\lambda)
=\mathsf{LB}(\F^{-}_{1,1},\Gamma)/\lambda$ and hence
$|R_i''(\lambda)|
=2\mathsf{LB}(\F^{-}_{1,1},\Gamma)/\lambda^3<A_i^3$.
Combining the bounds proves \eqref{eq:mhr-second-derivative}.

\end{proof}

Set $\Delta_\lambda=0.002$.  For each $i$, let
$\Lambda_i$ contain all multiples of $\Delta_\lambda$ in
$[0,1/A_{i-1}]$, together with the two points $1/A_i$ and
$1/A_{i-1}$.  Adjacent points of $\Lambda_i$ are at distance at most
$\Delta_\lambda$, and the kink of $R_i$ at $1/A_i$ is explicitly included.
Define
\begin{equation}
B_i^b
:=\frac{\Delta_\lambda^2}{8}
\left[A_i^2(h_i^b-h_0^b)+6A_i^3\right].
\label{eq:mhr-interpolation-buffer}
\end{equation}

We can now replace every continuous $\lambda$-constraint by finitely many
strengthened inequalities.  The next lemma shows that the buffer in
\eqref{eq:mhr-interpolation-buffer} exactly absorbs the worst interpolation
error between consecutive grid points.
\begin{lemma}
\label{lem:mhr-finite-lower}
If
\begin{equation}
L_i^b(\lambda)\geq R_i(\lambda)+B_i^b,
\qquad
\lambda\in\Lambda_i,\quad 1\leq i\leq N,
\label{cond:mhr-finite-lower}
\end{equation}
then $L_i^b(\lambda)\geq R_i(\lambda)$ for every
$\lambda\in[0,1/A_{i-1}]$ and every $i$.
\end{lemma}

\begin{proof}
Let $f=L_i^b-R_i$ and consider two adjacent points $x_0,x_1$ of
$\Lambda_i$ that lie on the same smooth piece of $R_i$.  By
Lemma~\ref{lem:mhr-interpolation}, $|f''|\leq M_i^b$ on this interval.
The error of linear interpolation therefore satisfies
\[
f(x)\geq
\frac{x_1-x}{x_1-x_0}f(x_0)
+\frac{x-x_0}{x_1-x_0}f(x_1)
-\frac{M_i^b(x_1-x_0)^2}{8}.
\]
Condition \eqref{cond:mhr-finite-lower}, together with
$x_1-x_0\leq\Delta_\lambda$, makes the right-hand side nonnegative.
The point $1/A_i$ separates the two smooth pieces, so these intervals
cover the entire domain.

\end{proof}

\paragraph{Extending the lower constraints beyond $\bar a$.}
The preceding two lemmas certify the lower constraints only up to the last
breakpoint $\bar a$.  We now show that the prescribed linear tail supplies
enough additional expectation for every $a>\bar a$, so no further grid in
$a$ is needed.  Specifically, Lemma~\ref{lem:mhr-cell-reduction} and
Lemma~\ref{lem:mhr-finite-lower} cover $a\leq\bar a$.  Let
$\lambda a\leq1$ and $a>\bar a$.  The linear tail in
\eqref{eq:mhr-step-function}, integration by parts, and
$e\,e^{-\lambda a}\geq1$ imply
\begin{align*}
\E_{\F^{-}_{\lambda,a}}[h^b(s)]
&\geq
\E_{\F^{-}_{\lambda,\bar a}}[h^b(s)]
+e\int_{\bar a}^{a}e^{-\lambda s}\,\dd s\\
&\geq
\mathsf{LB}(\F^{-}_{\lambda,\bar a},\Gamma)
+e\,e^{-\lambda a}(a-\bar a)\\
&\geq
\mathsf{LB}(\F^{-}_{\lambda,\bar a},\Gamma)
+(a-\bar a)\\
&\geq
\mathsf{LB}(\F^{-}_{\lambda,a},\Gamma),
\end{align*}
where the last inequality follows from
Lemma~\ref{lem:lb-lipschitz}.  Thus, the finite lower system certifies
\eqref{cond:mhr-lower} over the entire range $a\geq0$.

\paragraph{Covering the parameter $b$.}
It remains to cover the continuum of upper distributions with finitely many
witnesses.  We use stochastic monotonicity: a witness is evaluated under the
distribution at the right endpoint of an interval, while its allowable upper
benchmark is evaluated at the left endpoint.
For a function of the form \eqref{eq:mhr-step-function}, direct integration
under $\F^{+}_{1,b}$ gives
\begin{equation}
\E_{\F^{+}_{1,b}}[h^b(s)]
=h_0^b+\sum_{i=1}^{N}d_i^b q_i(b)
+e\,e^{-(\bar a-b)},
\qquad
q_i(b)=
\begin{cases}
1,&A_{i-1}\leq b,\\
e^{-(A_{i-1}-b)},&A_{i-1}>b.
\end{cases}
\label{eq:mhr-upper-expectation}
\end{equation}
Notice that the last term in \eqref{eq:mhr-upper-expectation} is the
expectation of the linear tail in \eqref{eq:mhr-step-function}; the level
$h_N^b$ is already included through the increment terms.

Let
\[
1=\beta_0<\beta_1<\cdots<\beta_M=10
\]
be an arbitrary partition.  For every $m$, suppose there exists a function
$h^m$ satisfying \eqref{cond:mhr-finite-lower}, $h_0^m\geq\Gamma$, and
$d_i^m\geq0$, together with
\begin{equation}
\E_{\F^{+}_{1,\beta_m}}[h^m(s)]
\leq
\mathsf{UB}(\F^{+}_{1,\beta_{m-1}},\Gamma).
\label{cond:mhr-finite-upper}
\end{equation}
For every $b\in(\beta_{m-1},\beta_m]$,
$\F^{+}_{1,\beta_m}$ first-order stochastically dominates
$\F^{+}_{1,b}$.  In addition, let
$U(b)=\mathsf{UB}(\F^{+}_{1,b},\Gamma)$.  For $0<\Gamma<1$,
implicit differentiation of the revenue equation following
\eqref{eq:mhr-ub-scalar} gives
\[
U'(b)=\frac{U(b)(b-1)}{b(U(b)-1)}\geq0,
\]
whereas $U(b)=b$ when $\Gamma=1$.  Thus, the upper benchmark is
non-decreasing in $b$.  Hence
\eqref{cond:mhr-finite-upper} implies
\[
\E_{\F^{+}_{1,b}}[h^m(s)]
\leq
\E_{\F^{+}_{1,\beta_m}}[h^m(s)]
\leq
\mathsf{UB}(\F^{+}_{1,\beta_{m-1}},\Gamma)
\leq
\mathsf{UB}(\F^{+}_{1,b},\Gamma).
\]
Thus, one witness at the right endpoint certifies the entire interval.
At the boundary point $b=1$, we use $h^1$ and the same chain of
inequalities with $\beta_0=1$.

\paragraph{The analytic extension for $b>10$.}
The finite partition will stop at $b=10$.  We dispose of the remaining range
$b>10$ analytically by constructing a universal linear witness, thereby
avoiding any truncation of the parameter space.

The next lemma bounds the slope required by all lower distributions.  Its
strict numerical margin ensures that the resulting linear witness also
satisfies the upper constraint whenever $b>10$.
\begin{lemma}
\label{lem:mhr-large-b}
Let $0<\Gamma\leq0.8$.  For $z\in[0,1]$, define
\[
\kappa_\Gamma(z):=
\begin{cases}
\displaystyle
\frac{\mathsf{LB}(\F^{-}_{1,z},\Gamma)}
     {1-e^{-z}},&z>0,\\[6pt]
\Gamma,&z=0,
\end{cases}
\qquad
c_\Gamma:=\max_{z\in[0,1]}\kappa_\Gamma(z).
\]
Then $c_\Gamma<0.9$.
\end{lemma}

\begin{proof}
Let $m(z)=1-e^{-z}$ and $p=0.9m(z)$.  Since $p\leq z\leq1$, the
revenue curve of $\F^{-}_{1,z}$ is increasing at $p$.  Moreover,
\begin{align*}
\log\frac{\Rev(p,\F^{-}_{1,z})}
          {0.8\,\OPT(\F^{-}_{1,z})}
&=
\log\frac{0.9}{0.8}
+\log\frac{1-e^{-z}}{z}
+z-0.9(1-e^{-z})\\
&\geq
\log\frac{9}{8}
+\frac{z}{2}
-0.9(1-e^{-z}).
\end{align*}
The inequality uses
\[
\frac{1-e^{-z}}{z}
=e^{-z/2}\frac{\sinh(z/2)}{z/2}
\geq e^{-z/2}.
\]
The last expression is minimized at $z=\log(9/5)$, where it equals
\[
\log\frac98+\frac12\log\frac95-\frac25
>0.0116.
\]
Therefore posting $0.9(1-e^{-z})$ obtains strictly more than
$0.8\,\OPT$, and hence strictly more than $\Gamma\,\OPT$.  It follows
that
$\mathsf{LB}(\F^{-}_{1,z},\Gamma)
<0.9(1-e^{-z})$ for $z>0$.  Moreover,
$\kappa_\Gamma(z)\to\Gamma$ as $z\downarrow0$, so
$\kappa_\Gamma$ is continuous on $[0,1]$.  Since
$\kappa_\Gamma(0)=\Gamma\leq0.8$, compactness gives
$c_\Gamma<0.9$.

\end{proof}

For $b>10$, set $h^b(s)=c_\Gamma s$.  Since
$\E_{\F^{+}_{1,b}}[s]=b+1$, Lemma~\ref{lem:mhr-large-b} gives
\[
\E_{\F^{+}_{1,b}}[h^b(s)]
=c_\Gamma(b+1)<0.9(b+1)<b
\leq\mathsf{UB}(\F^{+}_{1,b},\Gamma).
\]
For every $\lambda a\leq1$, scaling by $\lambda$ gives
\[
\E_{\F^{-}_{\lambda,a}}[h^b(s)]
=\frac{c_\Gamma}{\lambda}(1-e^{-\lambda a})
\geq
\frac{1}{\lambda}
\mathsf{LB}(\F^{-}_{1,\lambda a},\Gamma)
=\mathsf{LB}(\F^{-}_{\lambda,a},\Gamma),
\]
with the case $\lambda=0$ obtained by continuity.

The preceding reductions identify a finite set of verifiable conditions:
strengthened lower inequalities at the $\lambda$-grid points and one upper
inequality for each interval in the $b$-partition.  The following theorem
collects these conditions and shows that they suffice for a mechanism over
the entire MHR class.

\begin{theorem}
\label{thm:mhr-finite-certificate}
Let $0<\Gamma\leq0.8$.  Suppose there is a partition
$1=\beta_0<\cdots<\beta_M=10$ and, for each $m$, a function $h^m$ of
the form \eqref{eq:mhr-step-function}, with
$h_0^m\geq\Gamma$ and $d_i^m\geq0$, satisfying the following two
conditions, where $L_i^m$ and $B_i^m$ are obtained from
\eqref{eq:mhr-cell-lhs} and \eqref{eq:mhr-interpolation-buffer} by
replacing $h^b$ with $h^m$:
\begin{align}
L_i^m(\lambda)
&\geq R_i(\lambda)+B_i^m,
&&\lambda\in\Lambda_i,\quad
1\leq i\leq N,\quad 1\leq m\leq M,
\label{cond:mhr-certificate-lower}\\
\E_{\F^{+}_{1,\beta_m}}[h^m(s)]
&\leq
\mathsf{UB}(\F^{+}_{1,\beta_{m-1}},\Gamma),
&&1\leq m\leq M.
\label{cond:mhr-certificate-upper}
\end{align}
Then there exists a monotone sample-contingent mechanism that attains a
$\Gamma$-approximation over the class of MHR distributions.
\end{theorem}

\begin{proof}
The particular partition reported below is constructed by adaptive
bisection.  Starting from $[1,10]$, consider a candidate interval
$(\beta_-,\beta_+]$ and solve the finite linear program, subject to the
lower system and monotonicity constraints, for a witness at the right
endpoint $\beta_+$.  If the resulting witness satisfies
\eqref{cond:mhr-certificate-upper} with
$(\beta_{m-1},\beta_m)=(\beta_-,\beta_+)$, accept the interval; otherwise,
bisect it and test the two subintervals.  Processing the accepted intervals
from left to right produces the endpoints
$1=\beta_0<\cdots<\beta_M=10$ and a gap-free cover of $[1,10]$.
The verification below uses only the stored endpoints and witnesses, not the
solver that generated them.

For $b\in(\beta_{m-1},\beta_m]$, use the function $h^m$, and use $h^1$
at $b=1$.  For $b>10$, use the linear function from
Lemma~\ref{lem:mhr-large-b}; and rescale for general
$\F^{+}_{\lambda,b}$.  Lemmas~\ref{lem:mhr-cell-reduction}
and~\ref{lem:mhr-finite-lower}, together with the extension above
$\bar a$, establish every lower constraint
\eqref{cond:mhr-lower}.  The monotonicity argument following
\eqref{cond:mhr-finite-upper}, applied to
\eqref{cond:mhr-certificate-upper}, and the large-$b$ construction
establish the upper constraint \eqref{cond:mhr-upper}.  The functions are
non-decreasing by construction.  The characterization in
Section~\ref{sec:optimal} now yields the claimed approximation ratio.

\end{proof}

\begin{proposition}
\label{prop:optimal::dist}
For the class of MHR distributions, there exists a monotone sample-contingent
mechanism that attains a $0.791$-approximation ratio.
\end{proposition}

\paragraph{Upper bound on $\Gamma$.}
We finally establish the converse bound in
Proposition~\ref{prop: optimal numerical}.  Unlike the constructive lower
bound, this direction uses a finite restriction of the primal problem to
certify that a proposed ratio is impossible.
We retain the discretized primal certificate from the original computation.
Specifically, restricting program~\ref{lp:primal} to a finite support on
$(a,\lambda)$ and imposing the first-order stochastic-dominance constraints
with an adverse one-grid-point shift gives a finite linear program whose value
lower-bounds the value of program~\ref{lp:primal}.  Therefore, any parameters
$b$ and $\Gamma$ for which this finite value exceeds
$\mathsf{UB}(\F^{+}_{1,b},\Gamma)$ certify impossibility at
$\Gamma$.

At $\Gamma=0.796$ and $b=1.95$, the certified finite primal value is
greater than $2.3765$.  In comparison,
$\mathsf{UB}(\F^{+}_{1,b},\Gamma)\leq2.3756$.

\begin{proposition}
For the class of MHR distributions, no monotone concave pricing policy can achieve an
approximation ratio greater than $0.796$.
\end{proposition}

\section{Missing Proofs from Section~\ref{sec:application}}

\begin{proof}[Proof of Theorem~\ref{thm:scale}]
    The left-hand side is at least the right-hand side because the linear pricing rules \(\rho(\Psi(\F))=\omega\Psi(\F)\) form a subset of all feasible rules \(\rho(\Psi(\F))\).
    \begin{align*}
    \max_{\rho:\R \to \R}\, \min_{\theta} \, \min_{\F\in \cF(\theta)} \frac{\Rev(\rho(\theta),\F)}{\OPT(\F)}  =\max_{\rho:\R \to \R}\,  \min_{\F\in \cF} \frac{\Rev(\rho(\Psi(\F)),\F)}{\OPT(\F)} \ge \max_{\omega} \min_{\F \in \cF} \frac{\Rev(\omega\cdot \Psi(\F), \F)}{\opt(\F)}~.   
    \end{align*}
To show the reverse inequality, restrict nature to the normalized level set \(\cF(1)\):
\[
\max_{\rho:\R \to \R}\, \min_{\theta} \, \min_{\F\in \cF(\theta)} \frac{\Rev(\rho(\theta),\F)}{\OPT(\F)} \le \max_{\rho:\R \to \R}\, \min_{\F\in \cF(1)} \frac{\Rev(\rho(1),\F)}{\OPT(\F)} = \max_{\omega} \min_{\F \in \cF(1)} \frac{\Rev(\omega, \F)}{\opt(\F)}~.
\]   
Let \(\omega^*\) maximize the final expression.
For any \(\F\in\cF\) with \(\Psi(\F)>0\), define \(\F'\) by \(\F'(v)=\F(\Psi(\F)v)\). Homogeneity of \(\Psi\) and scale invariance of \(\cF\) imply that \(\F'\in\cF(1)\).
For every such \(\F\), homogeneity of revenue and of \(\opt\) gives
\[
\frac{\Rev(\omega^*\cdot \Psi(\F), \F)}{\opt(\F)}
= \frac{\omega^* \Psi(\F) \cdot \barf(\omega^* \Psi(\F))}{\opt(\F)}
= \frac{\omega^* \cdot \barf[\F'](\omega^*)}{\opt(\F')}
\ge \min_{\F' \in \cF(1)} \frac{\Rev(\omega^*, \F')}{\opt(\F')}~.
\]
Taking the minimum over \(\F\in\cF\) yields one inequality; the reverse follows from \(\cF(1)\subseteq\cF\). Hence the normalized and global linear-pricing problems have the same value. Combining this identity with the preceding bounds proves the equivalence and shows that \(\rho^*(\theta)=\omega^*\theta\):
\[
\max_{\rho:\R \to \R}\, \min_{\theta} \, \min_{\F\in \cF(\theta)} \frac{\Rev(\rho(\theta),\F)}{\OPT(\F)} = 
\max_{\omega} \min_{\F \in \cF} \frac{\Rev(\omega\cdot \Psi(\F), \F)}{\opt(\F)}
= \min_{\F \in \cF} \frac{\Rev(\omega^*\cdot \Psi(\F), \F)}{\opt(\F)}
\]

\end{proof}

\begin{proof}[Proof of Theorem~\ref{thm: scale-free}]
Since $\omega \Psi(\F)$ is monotone in $\F$, by \Cref{thm:monotone}, we have that for any $\omega$,
\begin{align*}
\min_{\F \in \cF} \frac{\Rev(\omega\cdot \Psi(\F), \F)}{\opt(\F)} 
 =\min_{\F\in \scF\cup \lcF} \frac{\Rev(\omega\cdot \Psi(\F),\F)}{\OPT(\F)}
\end{align*}
By \eqref{eq:worstF-set}, every \(\F\in\scF\cup\lcF\) is of the form \(\sF\) or \(\lF\).
If $\F = \sF$, let $\F' = \sFone$. We then have $\Psi(\F) = a\Psi(\F')$ and $\opt(\F)=a\opt(\F')$. If $\F = \lF$, let $\F' = \lFone$. We then have $\Psi(\F) = b\Psi(\F')$ and $\opt(\F)=b\opt(\F')$. In both cases, we have that $\F' \in \scFone \cup \lcFone$ and
\begin{align*}
 \frac{\Rev(\omega\cdot \Psi(\F),\F)}{\OPT(\F)}
&= \frac{\omega\Psi(\F) \cdot \barf(\omega\Psi(\F))}{\OPT(\F)} \\
&= \frac{\omega\Psi(\F') \cdot \barf[\F'](\omega\Psi(\F'))}{\OPT(\F')}
 = \frac{\Rev(\omega\cdot \Psi(\F'),\F')}{\OPT(\F')}~.
\end{align*}
Since $\scFone \subseteq \scF$ and $\lcFone\subseteq \lcF$, for any $\omega$, we have that 
$$
\min\limits_{\F \in\scFone \cup \lcFone} \frac{\Rev(\omega\cdot \Psi(\F), \F)}{\opt(\F)} = \min\limits_{\F \in \scF\cup \lcF} \frac{\Rev(\omega\cdot \Psi(\F), \F)}{\opt(\F)} =  \min_{\F \in \cF} \frac{\Rev(\omega\cdot \Psi(\F), \F)}{\opt(\F)}~,
$$ 
which completes the proof of \eqref{eq: scale ab}.

Next, denote $\theta =\Psi(\F)$ for any $\F\in \scF\cup \lcF$ in the form of $\sF$ or $\lF$. Define $\F'$ such that $\F'(v) = \F(\theta v)$ for all $v$: if $\F = \sF$ then let $\F' = \F^{-}_{\lambda \theta, a/\theta}$; if $\F = \lF$ then let $\F' = \F^{+}_{\lambda \theta, b/\theta}$. Since $\Psi$ is homogeneous, we have that $\Psi(\F')=1$ and thus, $\F'\in \scF(1) \cup \lcF(1)$. Moreover, 
\[
\frac{\Rev(\omega\cdot \Psi(\F),\F)}{\OPT(\F)} 
= \frac{\omega\theta \cdot \barf(\omega\theta)}{\OPT(\F)} = \frac{ \omega \theta \Psi(\F') \cdot \barf[\F'](\omega\Psi(\F'))}{\theta \OPT(\F')} = \frac{\Rev(\omega\cdot \Psi(\F'),\F')}{\OPT(\F')}~.
\]
Since $\scF(1) \subseteq \scF$ and $\lcF(1)\subseteq \lcF$, for any $\omega$, we have that 
$$
\min\limits_{\F \in\scF(1) \cup \lcF(1) } \frac{\Rev(\omega\cdot \Psi(\F), \F)}{\opt(\F)} = \min\limits_{\F \in \scF\cup \lcF} \frac{\Rev(\omega\cdot \Psi(\F), \F)}{\opt(\F)} =  \min_{\F \in \cF} \frac{\Rev(\omega\cdot \Psi(\F), \F)}{\opt(\F)}~,
$$ 
which completes the proof of \eqref{eq: scale psi}.

\end{proof}

\begin{proof}[Proof of Proposition~\ref{prop: norm}]
According to \eqref{eq: scale ab} in \Cref{thm: scale-free}, without loss of generality, nature may restrict attention to $\{\F^{-}_{\lambda,1} : \lambda \le 1\}$ and $\{\F^{+}_{\lambda,1} : \lambda \ge 1\}$. 
We first calculate the $L^\eta$ norm of $\sF[1]$ and $\lF[1]$:
\begin{align*}
 \|\F^{-}_{\lambda,1}\|_\eta
 &= \of{\int_{0}^{1} v^\eta \lambda e^{-\lambda v}\dd v + e^{-\lambda}}^{1/\eta} \\
 &= \of{\int_{0}^{\lambda} \of{\frac{t}{\lambda}}^\eta e^{-t}\dd t + e^{-\lambda}}^{1/\eta}
  = \of{\lambda^{-\eta} \int_0^{\lambda} t^{\eta}e^{-t}\dd t + e^{-\lambda}}^{1/\eta}; \\
 \|\F^{+}_{\lambda,1}\|_\eta
 &= \of{\int_{1}^{\infty} v^\eta \lambda e^{-\lambda (v-1)}\dd v}^{1/\eta} \\
 &= \of{e^\lambda \int_{\lambda}^{\infty} \of{\frac{t}{\lambda}}^\eta e^{-t}\dd t}^{1/\eta}
  = \of{e^{\lambda}\lambda^{-\eta} \int_{\lambda}^\infty t^{\eta}e^{-t}\dd t}^{1/\eta}.
\end{align*}
\begin{itemize}
    \item For distribution $\sF[1]$, the purchase probability under price  $\omega \|\sF[1]\|_\eta$ is $e^{-\lambda \omega \|\sF[1]\|_\eta}$, resulting in a revenue of  $\omega \cdot \|\sF[1]\|_\eta \cdot e^{-\lambda \omega \|\sF[1]\|_\eta}$. Notice that $\opt(\sF[1]) = e^{-\lambda}$. This corresponds to constraint \eqref{eq:norm-a}.
    \item For distribution $\lF[1]$, the purchase probability under price  $\omega \|\lF[1]\|_\eta$ is $\min \{ e^{-\lambda (\omega \|\lF[1]\|_\eta-1)}, 1\}$, resulting in a revenue of  $\omega \cdot \|\lF[1]\|_\eta \cdot \min \left\{ e^{-\lambda(\omega \|\lF[1]\|_\eta-1)}, 1\right\}$. Notice that $\opt(\lF[1]) = 1$. This corresponds to constraint \eqref{eq:norm-b}.
\end{itemize}

Finally, optimizing over the discount factor  $\omega\in [0,1]$ yields the optimal approximation ratio $\Gamma$ for the pricing rule. This completes the proof.

\end{proof}

\begin{proof}[Proof of Proposition~\ref{prop:superquantile}]
By \eqref{eq: scale psi} in \Cref{thm: scale-free}, evaluating worst-case performance subject to \(\mathrm{CVaR}_{q}=1\) reduces to the boundary distributions \(\sF\) and \(\lF\) in \eqref{eq:worstF}. We therefore compute their superquantiles as functions of the boundary-distribution parameters. For \(u\in(0,1]\), their upper-tail VaRs are
\begin{align*}
\mathrm{VaR}_{u}(\sF)
&=
\begin{cases}
 a, & 0<u\le e^{-\lambda a},\\
 -\dfrac{\ln u}{\lambda}, & e^{-\lambda a}<u\le1,
\end{cases}\\
\mathrm{VaR}_{u}(\lF)
&=b-\dfrac{\ln u}{\lambda}.
\end{align*}
When \(\lambda=0\), \(\sF\) is degenerate at \(a\), so \(\mathrm{VaR}_{u}(\sF)=a\) for every \(u\in(0,1]\).
\begin{enumerate}[nosep]
    \item 
For $\sF$, we calculate the $\mathrm{CVaR}_{q}$ as follows.
\begin{enumerate}[nosep]
    \item If $q\le e^{-\lambda a}$, then 
\begin{align*}
\mathrm{CVaR}_{q}(\sF)=\frac{1}{q}\int_0^q \mathrm{VaR}_{u}(\sF)\,\dd u = a.
\end{align*}
So feasible $(\lambda, a)$ need to satisfy $a=1$, $\lambda\le -\ln q$. For parameter $\lambda$, the purchase probability under price $\omega$ is $e^{-\lambda \omega}$, and thus, it obtains a revenue of  $\omega e^{-\lambda \omega}$. The optimal revenue under this distribution is $a e^{-\lambda a}=e^{-\lambda}$.
Hence, the performance ratio for price $\omega$ is $\frac{\omega\cdot e^{-\lambda \omega}}{e^{-\lambda }}=\omega  e^{\lambda\cdot\of{1-\omega}}$ which is increasing in $\lambda$. Therefore, the worst-case distribution is achieved at $\lambda=0$, and  $\omega e^{\lambda\cdot\of{1-\omega}}$ is lower bounded by $\omega$. This corresponds to constraint \eqref{eqn:cvar_1}.
\item If $q> e^{-\lambda a}$,
\begin{align*}
\mathrm{CVaR}_{q}(\sF)
&=\frac{1}{q}\int_0^q \mathrm{VaR}_{u}(\sF)\,\dd u \\
&= \frac{1}{q} \of{\int_0^{e^{-\lambda a}} a\dd q'
  + \int_{e^{-\lambda a}}^{q} -\frac{\ln(q')}{\lambda}\dd q'} \\
&= \frac{1}{\lambda}\cdot\of{1-\ln q-\frac{e^{-\lambda a}}{q}}.
\end{align*}
which implies that $e^{-\lambda a}=-q\cdot\of{\lambda+\ln q-1 }$. Since $0<e^{-\lambda a}<q$, 
 feasible $(\lambda, a)$ need to satisfy $-\ln q<\lambda <1  -\ln q$ and $a=-\frac{\ln\of{-q\cdot\of{\lambda+\ln q-1 }}}{\lambda}$. Since $\lambda a\le 1$, it follows that $\lambda\le 1-\ln q -\frac{1}{eq}$. For parameter $\lambda$, price $\omega$ obtains a revenue of  $\omega e^{-\lambda \omega}$. The optimal revenue under this distribution is $a e^{-\lambda a}=-\frac{\ln\of{-q\cdot\of{\lambda+\ln q-1 }}}{\lambda}\cdot  q\cdot\of{1-\lambda-\ln q}$.
Hence, the performance ratio for price $\omega$ is $\frac{\omega\cdot e^{-\lambda \omega}}{q\cdot\of{1-\lambda-\ln q}\cdot -\frac{\ln\of{-q\cdot\of{\lambda+\ln q-1 }}}{\lambda} }$,
 which corresponds to constraint \eqref{eqn:cvar_2}.
 \end{enumerate}
\item For $\lF$, $\mathrm{CVaR}_{q}$ is calculated as follows
\begin{align*}
\mathrm{CVaR}_{q}(\lF)=\frac{1}{q}\int_0^q \mathrm{VaR}_{u}(\lF)\,\dd u=\frac{1}{q}\int_0^q\left(b-\frac{\ln u}{\lambda}\right)\dd u = b+\frac{1}{\lambda}\of{1-\ln q}.
\end{align*}
Since $b+\frac{1}{\lambda}\of{1-\ln q} =1$ and membership in \(\lcF\) requires \(\lambda b\ge1\), feasibility requires
\[
\lambda\ge2-\ln q,
\qquad
b=1-\frac{1}{\lambda}\of{1-\ln q}.
\]
The purchase probability at price $\omega$ is
\[
\min\offf{1, e^{-\lambda(\omega-b)}}
= \min\offf{1, e^{-\lambda \omega+\lambda-1+\ln q}},
\]
so price $\omega$ achieves revenue
$\omega\cdot \min\offf{1, q e^{-\lambda \of{\omega-1+1/\lambda}}}$.
The optimal revenue under this distribution is attained at \(b\), because \(\lambda b\ge1\), and equals\ 
$b=1-\frac{1}{\lambda}\of{1-\ln q}$.
Therefore, the performance ratio for price $\omega$ is
\[
\frac{\omega\cdot \min\offf{1, q e^{-\lambda \of{\omega-1+1/\lambda}}}}
{1-\frac{1}{\lambda}\of{1-\ln q}},
\]
which corresponds to constraint \eqref{eqn:cvar_3}.
\end{enumerate}
Hence, problem~\eqref{eq:cvar} determines the optimal price \(\omega\) and approximation ratio.

\end{proof}

\end{document}